\documentclass[a4paper,11pt]{article}
\usepackage{ifthen,calc}
\usepackage{amsmath,amssymb,amsthm,amstext,amsfonts,bm,mathrsfs,commath,comment}
\usepackage[scale=0.76]{geometry}
\usepackage{color,graphicx,pgf,tikz}
\usetikzlibrary{calc,trees,positioning,arrows,chains,shapes,matrix,automata,patterns,shapes.geometric,decorations.pathreplacing,decorations.pathmorphing,decorations.markings,shapes.symbols}
\tikzset{>=latex}
\usepackage{placeins}	
\usepackage{graphicx}
\usepackage{float}
\usepackage{subfigure}
\usepackage{extarrows}

\usepackage{natbib}
\usepackage{enumitem}
\setlist{nosep}

\usepackage{setspace}
\allowdisplaybreaks

\usepackage{lineno}

\newtheorem{defn}{Definition}
\newtheorem{thm}{Theorem}
\newtheorem{lem}{Lemma}

\newtheorem{prop}{Proposition}
\newtheorem{rmk}{Remark}

\newtheorem{claim}{Claim}

\newcommand{\bP}{\mathbb{P}}
\newcommand{\bQ}{\mathbb{Q}}
\newcommand{\bR}{\mathbb{R}}
\newcommand{\bZ}{\mathbb{Z}}

\newcommand{\cF}{\mathcal{F}}

\newcommand{\cI}{\mathcal{I}}

\newcommand{\cM}{\mathcal{M}}

\newcommand{\cP}{\mathscr{P}}

\newcommand{\cU}{\mathcal{U}}

\newcommand{\cW}{\mathcal{W}}

\DeclareMathOperator{\supp}{supp}
\DeclareMathOperator{\Br}{Br}

\DeclareMathOperator*{\argmax}{arg\,max}

\usepackage{hyperref}
\hypersetup{colorlinks=true,
            citecolor=blue,
            linkcolor=blue,
            urlcolor=blue,
            bookmarksopen=true,
            pdfstartview={XYZ null null 1.50},
            pdfpagelayout={SinglePage}
}

\usepackage{scalerel,stackengine}
\stackMath
\newcommand\ehat[1]{%
\savestack{\tmpbox}{\stretchto{%
  \scaleto{%
    \scalerel*[\widthof{\ensuremath{#1}}]{\kern-.6pt\bigwedge\kern-.6pt}%
    {\rule[-\textheight/2]{1ex}{\textheight}}
  }{\textheight}%
}{0.5ex}}%
\stackon[1pt]{#1}{\tmpbox}%
}

\begin{document}
\title{Robust perfect equilibrium in large games\thanks{The authors are grateful to Wei He, Xiao Luo, and Satoru Takahashi for helpful discussions. Part of the results in this paper were presented at the 18th SAET Conference, Taipei, June 11--13, 2018, the 4th PKU-NUS Annual International Conference on Quantitative Finance and Economics, Shenzhen, May 11--12, 2019, and the 12th World Congress of the Econometric Society, August 17--21, 2020; and we thank the participants for their constructive comments.
Xiang Sun's research is supported by NSFC (No. 71773087), Fok Ying Tong Education Foundation (No. 171076), and Youth Scholar Team in Social Sciences of Wuhan University ``New perspectives for development economics research''.
Enxian~Chen and Yeneng Sun acknowledge the financial support from NUS Grant R-146-000-286-114.
}}
\author{Enxian~Chen\thanks{Department of Mathematics, National University of Singapore, 10 Lower Kent Ridge Road, Singapore, 119076. E-mail: \href{mailto:e0046840@u.nus.edu}{e0046840@u.nus.edu}.}
\and
Lei~Qiao\thanks{School of Economics, Shanghai University of Finance and Economics, 777 Guoding Road, Shanghai, 200433, China. E-mail: \href{mailto:qiao.lei@mail.shufe.edu.cn}{qiao.lei@mail.shufe.edu.cn}.}
\and
Xiang~Sun\thanks{Center for Economic Development Research, Economics and Management School, Wuhan University, 299 Bayi Road, Wuhan, 430072, China. E-mail: \href{mailto:xiangsun.econ@gmail.com}{xiangsun.econ@gmail.com}.}
\and
Yeneng~Sun\thanks{Departments of Economics and Mathematics, National University of Singapore, Singapore, 119076. E-mail: \href{mailto:ynsun@nus.edu.sg}{ynsun@nus.edu.sg}.}
}
\date{This version: \today}
\maketitle

\begin{abstract}
This paper proposes a new equilibrium concept ``robust perfect equilibrium'' for non-cooperative games with a continuum of players, incorporating three types of perturbations. Such an equilibrium is shown to exist (in symmetric mixed strategies and in pure strategies) and satisfy the important properties of admissibility, aggregate robustness, and ex post robust perfection. These properties strengthen relevant equilibrium results in an extensive literature on strategic interactions among a large number of agents. Illustrative applications to congestion games and potential games are presented. In the particular case of a congestion game with strictly increasing cost functions, we show that there is a unique symmetric robust perfect equilibrium.


\end{abstract}

\newpage
\tableofcontents
\setlength{\parskip}{2pt}


\section{Introduction}
\label{sec:intro}

In the study of interactions with a large number of agents, it is a common behavioral assumption that an individual agent chooses an optimal action with negligible effect on the aggregate. One classical example involves the idea of perfect competition. Namely, a large population of buyers and sellers ensures that supply meets demand in the market, while each individual agent cannot influence the market prices. Another classical example concerns strategic interactions among a large number of agents, which is the focus of this paper.\footnote{Economies and games with a large number of agents are also called large economies and large games respectively, which are surveyed, for example, in \cite{Hildenbrand1974}, \cite{Anderson1994} and \cite{KS2002}. For recent developments, see \cite{MP2002}, \cite{Kalai2004}, \cite{KP2009}, \cite{Yannelis2009}, \cite{Yu2014}, \cite{CP2014, CP2020}, \cite{BH2015}, \cite{Hammond2015}, \cite{HY2016}, \cite{KS2018}, \cite{Hellwig2019} and \cite{Yang2021} among others.} A particular case is the study of transportation. Travelers choose routes that they perceive to have the least cost, under the prevailing traffic conditions that cannot be influenced by an individual traveler. Such an idea of a congestion game had already been considered in \cite{Pigou1920} and \cite{Knight1924},\footnote{Some papers trace back the literature on congestion games as early as to the work of \cite{Kohl1841}.} and was formalized as an equilibrium concept in \cite{Wardrop1952}, the so-called Wardrop equilibrium.\footnote{There is a large literature on Wardrop equilibrium in the study of transportation and telecommunication networks. Google Scholar shows that \cite{Wardrop1952} has been cited in over six thousand articles. See \cite{BMW1956} for an early mathematical formulation of Wardrop equilibrium, and \citet[Section 1.4.5]{FP2003} for a modern treatment.} It is clear that Wardrop equilibrium is simply Nash equilibrium for games with a continuum of players.\footnote{See, for example, Page 5 of the lecture notes of Acemoglu and Ozdaglar at \url{https://economics.mit.edu/files/4712}.} Despite wide applications of large games in various areas,\footnote{Besides the applications in transportation and telecommunication networks, other examples include \cite{AJ2015} in macroeconomics, \cite{CMP1992} in public economics, \cite{DGG2005} in finance, \cite{Peters2010} in labor economics, \cite{Sandholm2015} in population games, \cite{OS2016} in contest games, \cite{HMC2006}, \cite{GLL2011}, \cite{CD2018} in mean-field games.} the existing literature on such games hardly goes beyond Nash equilibrium. The main purpose of this paper is to introduce and to study a new equilibrium concept (called ``robust perfect equilibrium'') which incorporates three types of perturbations. Such an equilibrium is shown to have the important properties of admissibility, aggregate robustness, and ex post robust perfection. These properties strengthen relevant equilibrium results in an extensive literature on large games. As a particular example, we obtain a unique symmetric robust perfect equilibrium in a congestion game with strictly increasing cost functions, going beyond the classical work on Wardrop equilibrium.

It is well known that a Nash equilibrium in a finite game is not necessarily admissible in the sense that players could choose weakly dominated strategies with positive probability.\footnote{See, for example, \citet[p.~248]{OR1994}. \cite{BFK2008} provide a characterization of admissibility.} The following is a variation of the classical example in \cite{Pigou1920}, which provides a simple congestion game showing the failure of admissibility of a Nash/Wardrop equilibrium. Drivers travel from the origin node $o$ to the destination node $t$ through two different Paths $a$ and $b$.
\begin{figure}[htb!]
\centering
\includegraphics{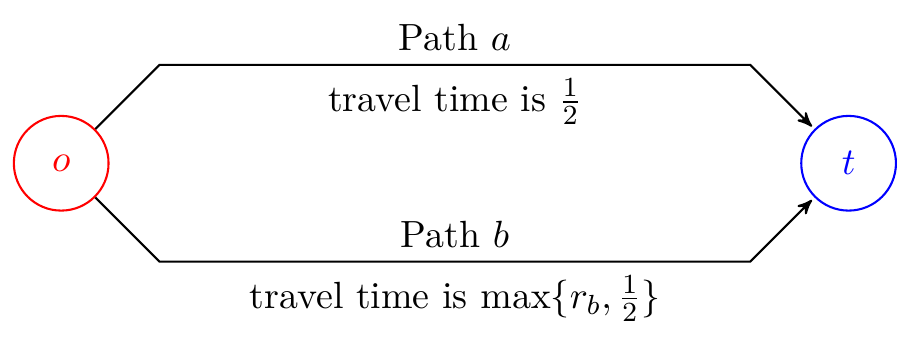}
\caption{Modified Pigou's game}
\label{fig:intro}
\end{figure}
Let $r_b$ denote the proportion of drivers choosing Path $b$. Path $a$ is a broad highway on which the transit time is always half an hour with no congestion. Path $b$ is a relatively narrow roadway on which (1) there is no congestion if no more than half of the drivers choose this roadway, and in this case the travel time is also fixed at half an hour; (2) there is congestion if at least half of the drivers choose it ($r_b \ge \frac{1}{2}$), and the travel time is $r_b$ hour. As in \cite{Pigou1920}, it is assumed that an individual driver has no influence on the proportions of the drivers on Paths $a$ and $b$, i.e., an infinite number of drivers is assumed implicitly. It is easy to identify a Nash equilibrium (or a Wardrop equilibrium) in this game (denoted by $f$): each path is taken by exactly half of the drivers. Since Path $a$ is as good as Path $b$ when $r_b\le \frac{1}{2}$ and strictly better when $r_b > \frac{1}{2}$ (Path $b$ is weakly dominated by Path $a$), there is little reason for a driver to choose Path $b$ over Path $a$. Thus, the above equilibrium, which violates admissibility, is not desirable from a behavioral point of view.

It is a widely adopted behavioral assumption that players could make mistakes with small probabilities. To reflect such an idea in a finite game, \cite{Selten1975} introduced the concept of trembling hand perfection. In particular, the strategy taken by each player is required to be a best response to a sequence of full-support perturbations on the other players' strategies.\footnote{For such an idea of perfection in the dynamic setting, see \cite{KW1982} and \cite{MR2020}.} Since the externality component in a large game usually depends on the aggregate action distributions (instead of the action profiles for all the other agents), a natural property (called ``aggregate robustness'') would require each player to choose a best response to a sequence of full-support perturbations on the aggregate action distributions. In the particular case of modified Pigou's game, any full-support perturbation on the aggregate will assign positive probability to the set of action distributions with more drivers on Path $b$ than on Path $a$. Thus, the expected travel time on Path $b$ with respect to any full-support perturbation will be more than $\frac{1}{2}$ hour. Since the expected travel time on Path $a$ is always $\frac{1}{2}$ hour, we know that a mixed strategy with positive probability on Path $b$ cannot be a best response to a sequence of full-support perturbations on the aggregate. In particular, the equilibrium $f$ in the modified Pigou's game, which requires half of the drivers to take Path $b$, violates the property of aggregate robustness. On the other hand, $f$ satisfies the idea of trembling hand perfection using full-support perturbations on the strategies of individual players as in \cite{Selten1975}, which means that the undesirable equilibrium $f$ is ruled out by aggregate robustness, but not by the usual idea of trembling hand perfection.\footnote{One can take a sequence of perturbed strategy profiles $\{f^n\}_{n=1}^\infty$ as follows: for any positive integer $n$, those drivers who choose Path $a$ in equilibrium $f$ choose Path $a$ with probability $1 - \frac{1}{n}$ and Path $b$ with probability $\frac{1}{n}$, while those drivers who choose Path $b$ in equilibrium $f$ choose Path $b$ with probability $1 - \frac{1}{n}$ and Path $a$ with probability $\frac{1}{n}$. The law of large numbers suggests that each path is taken by exactly half of the drivers. Thus, Paths $a$ and $b$ remain best responses to the sequence of perturbed strategy profiles.}

In the modified Pigou's example, Path $a$ is always a best response to any perturbation on the aggregate action distribution since Path $a$ weakly dominates Path $b$. Hence, the unique strategy profile satisfying admissibility and aggregate robustness is that all the drivers choose Path $a$. By adding a third Path $c$, we have a new congestion game with a non-dominant mixed strategy profile satisfying admissibility and aggregate robustness.
\begin{figure}[htb!]
\centering
\includegraphics{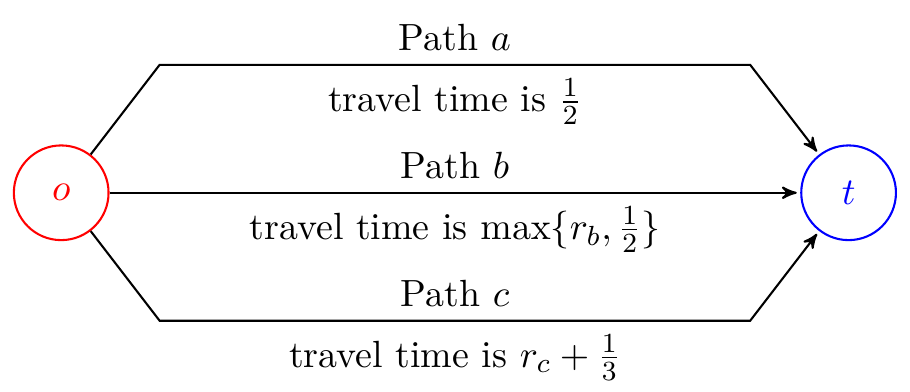}
\caption{Three-path game}
\label{fig:example}
\end{figure}
Let $r_c$ denote the proportion of drivers choosing Path $c$. Path $c$ is assumed to be a narrow shortcut with travel time $(r_c + \frac{1}{3})$ hour. Consider the symmetric strategy profile $g^0$ in which every driver uses the same mixed strategy $\sigma^0$ such that Paths $a$ and $c$ are chosen with probabilities $\frac{5}{6}$ and $\frac{1}{6}$ respectively. Intuitively, the law of large numbers suggests that the realized traffic flow induced by $g^0$ is $\tau^0 = (\frac{5}{6}, 0, \frac{1}{6})$, which means that $\frac{5}{6}$ (resp. $\frac{1}{6}$) of drivers choose Path $a$ (resp. $c$).\footnote{Such a claim is guaranteed by the exact law of large numbers as shown in \cite{Sun2006}, which is reproduced as Lemma~\ref{lem:ELLN} below.} It is easy to see that $g^0$ is a Nash/Wardrop equilibrium.\footnote{If one particular driver takes path $p$ with all the other drivers taking the mixed strategy $\sigma^0$, then her expected travel time is always $\frac{1}{2}$ no matter what $p$ is. Hence, the mixed strategy $\sigma^0$ is the best response of this driver when all the other drivers take the mixed strategy $\sigma^0$, i.e., the strategy profile $g^0$ is a Nash/Wardrop  equilibrium.} Since neither Path $a$ nor Path $c$ is weakly dominated, $g^0$ satisfies admissibility. We can also verify that $g^0$ satisfies aggregate robustness.\footnote{Let $\varepsilon_n = \frac{1}{6n}$ for any positive integer $n$. Consider a full-support distribution on the simplex $\Delta$ in $\bR^3$: $\ehat{\tau^n} = (1-\varepsilon_n)\delta_{(x_n, y_n, z_n)} + \varepsilon_n \eta$, where $x_n = \frac{5}{6}$, $y_n=\tfrac{\varepsilon_n}{6-6\varepsilon_n}$, $z_n = \frac{1-2\varepsilon_n}{6-6\varepsilon_n}$, $\delta_{(x_n, y_n, z_n )}$ is the Dirac measure at $(x_n, y_n, z_n)$, and $\eta$ is the uniform distribution on $\Delta$. The sequence $\{\ehat{\tau^n}\}_{n \in \bZ_+}$ provides full-support perturbations to the traffic flow $\tau^0$. Fix any perturbation $\ehat{\tau^n}$. The expected travel time for a driver to choose Path $a$ is $\tfrac{1}{2}$ hour. It is straightforward to check that if a driver chooses Path $c$, then her expected travel time is still $\tfrac{1}{2}$ hour. The same argument as in the modified Pigou's example implies that the expected travel time on Path $b$ with respect to any full-support perturbation will be more than $\tfrac{1}{2}$ hour. Hence, both $a$ and $c$ are best responses with respect to each perturbation $\ehat{\tau^n}$; so is the mixed strategy $\sigma^0 = (\frac{5}{6}, 0, \frac{1}{6})$. This concludes that $g^0$ satisfies aggregate robustness.}

It will be shown in Subsection \ref{sec:example:revisit} that $\tau^0$ is the only aggregate action distribution induced by a Nash equilibrium with admissibility and aggregate robustness in the above three-path game. Given the aggregate action distribution $\tau^0$, every path is optimal for each driver (with the same expected travel time $\tfrac{1}{2}$ hour). This also means that any pure strategy profile with $\frac{5}{6}$ of drivers choosing Path $a$ and $\frac{1}{6}$ of drivers choosing Path $c$ will be a Nash equilibrium, which will lead to a continuum of equilibrium satisfying admissibility and aggregate robustness. Thus, each individual driver faces the coordination problem with other drivers to play any such pure strategy equilibrium, unless the equilibrium action for every player is specified  by a social planner. On the other hand, if every driver takes the same mixed strategy $\sigma^0$, then such a symmetric Nash equilibrium is simple to ``implement'' without the coordination problem.\footnote{It is obvious that the three-path game has no symmetric pure strategy Nash equilibrium.} That is, every individual driver can roll (independently from other drivers) a six-sided dice to determine which path to take: Path $c$ if the outcome is one, and Path $a$ otherwise. As mentioned above, the law of large numbers suggests that the realized traffic flow of $g^0$ is still $\tau^0 = (\frac{5}{6}, 0, \frac{1}{6})$, which implies that the realized pure strategy profile is also a Nash equilibrium satisfying admissibility and aggregate robustness. Thus, the Nash equilibrium property as well as the properties of admissibility and aggregate robustness on $g^0$ are preserved after every driver takes the dice-determined path.

Since large games have wide applications, a natural question is whether one can formulate a new equilibrium concept refining Nash equilibrium and satisfying the desirable properties of admissibility and aggregate robustness. When one works with a mixed strategy profile in a large game, each player chooses a pure strategy according to some randomization device (independently across players) as specified by her mixed strategy (like rolling a dice as above). As discussed in the above paragraph, it will be useful for the new equilibrium in mixed strategies (especially in symmetric mixed strategies) to have the ex post preservation property in the sense that the realized pure strategy profile conforms to the same equilibrium concept. Working with such a symmetric mixed strategy equilibrium in a general large game can alleviate the coordination difficulty by individual players, and generate pure strategy profiles with the same equilibrium properties after realizations of the mixed strategies.

In this paper, we introduce the notion of robust perfect equilibrium by considering the following perturbations: (i) perturbations on the player space---we allow a small number of players to be ``irrational''; (ii) perturbations in players' strategies---this captures the idea that every player may ``tremble'' and fail to choose a best response; (iii) perturbations in aggregate action distributions---it allows that the reported societal summaries may be somewhat inaccurate. Theorem~\ref{thm:robust} shows that any robust perfect equilibrium in a general large game is a Nash equilibrium and satisfies the desirable properties of admissibility, aggregate robustness and ex post preservation (called ``ex post robust perfection'' in this particular setting). Theorem~\ref{thm:exist} establishes the existence of robust perfect equilibrium in symmetric mixed strategies (i.e., players with the same payoff take the same strategy) and also in pure strategies.

As illustrative applications of our main results, we consider general atomless congestion games and potential games.\footnote{In addition to several classical references mentioned earlier, recent applications of congestion games include \cite{Milchtaich2000, Milchtaich2004}, \cite{RT2002}, \cite{AO2007}, \cite{FILRS2016}, and \cite{AMMO2018}. Potential games are also widely studied in economics; see the references in \cite{CM2014} (which considers perfect equilibrium in the setting of finite-player potential games).} It follows from our Theorem~\ref{thm:exist} that pure strategy robust perfect equilibria exist in such games. In the specific setting of congestion games, our results indicate that the large literature on Wardrop (Nash) equilibrium can be strengthened to robust perfect equilibrium with the important properties of admissibility, aggregate robustness and ex post robust perfection. We also propose an optimization method for finding $\varepsilon$-robust perfect equilibria. In the case that the cost functions are strictly increasing in a congestion game, we show that there is a unique symmetric robust perfect equilibrium and the traffic flow induced by any (possibly non-symmetric) robust perfect equilibria is unique; there is no coordination problem at all for the drivers since the drivers can follow the unique symmetric  robust perfect equilibrium. Those results enrich the classical result on the uniqueness of Wardrop equilibrium as in \citet[Section~3.1.3, p.~3.13]{BMW1956}. It means that the unique Wardrop equilibrium satisfies not only the individual optimality condition as in the literature, but also the additional important properties of admissibility and aggregate robustness as considered in this paper. Furthermore, we show that an admissible Nash equilibrium must be robust perfect in any two-path congestion game.

The rest of the paper is organized as follows. In Section~\ref{sec:mixPE}, we present a model for large games, formulate the two properties and robust perfect equilibrium rigorously, and state various results about robust perfect equilibrium. Section~\ref{sec:example} revisits the three-path game and provides two additional examples to clarify the concepts of admissibility, aggregate robustness, and robust perfect equilibrium. Section~\ref{sec:application} considers robust perfect equilibrium in general atomless congestion games and potential games. All the proofs are collected in Section~\ref{sec:appen}.

\section{Robust perfect equilibria}
\label{sec:mixPE}

The modified Pigou's example in Section~\ref{sec:intro} shows that a Nash/Wardrop equilibrium may not satisfy any of the properties of admissibility and aggregate robustness. However, the three-path game in Section~\ref{sec:intro} also indicates that both properties can be satisfied nontrivially. In this section, we introduce and study the concept of robust perfect equilibrium with the two desirable properties. Some notations and basic definitions about large games are presented in Subsection~\ref{sec:mixPE:game}, while the admissibility and aggregate robustness are formally stated in Subsection~\ref{sec:mixPE:properties}. In Subsection~\ref{sec:mixPE:robustPE}, we introduce the notion of robust perfect equilibrium. Such an equilibrium is shown to exist with the properties of admissibility and aggregate robustness. Moreover, a robust perfect equilibrium satisfies ex post robust perfection in the sense that almost every realized pure strategy profile is a robust perfect equilibrium.

\subsection{Some basic definitions}
\label{sec:mixPE:game}

In this subsection, we introduce some notations and definitions for large games. Let $(I, \cI, \lambda)$ be an atomless probability space denoting the set of players,\footnote{Throughout this paper, we follow the convention that a probability space is complete and countably additive.} and $A = \{ a_1, a_2, \ldots, a_K \}$ a finite set representing the common action set for each player.\footnote{In this paper, we focus on large games with finite actions for simplicity. All the main results can be generalized to large games with general heterogeneous (and possibly infinite) action sets.} Let $\Delta$ denote the set of probability measures on $A$, which can be identified with the unit simplex in $\bR^K$
$$ \Delta = \{(x_1, x_2, \ldots, x_K) \mid x_1 + x_2 + \cdots + x_K = 1, x_k \ge 0, \forall k = 1, 2, \dots, K\}, $$
where $x_k$ in the vector $(x_1, x_2, \ldots, x_K)$ represents the probability at $a_k$. Given an action profile for the agents, the action distribution which specifies the proportions of agents taking the individual actions in $A$ (also called a societal summary) can be viewed an element in $\Delta$. Each player's payoff is a continuous function on $A \times \Delta$, which means that her payoff depends on her own actions and societal summaries continuously. Let $\cU$ be the space of continuous functions on $A \times \Delta$ endowed with the sup-norm topology and the resulting Borel $\sigma$-algebra.

A large game $G$ is a measurable mapping from $(I, \cI, \lambda)$ to $\cU$. A \emph{pure strategy profile} $f$ is a measurable function from $(I, \cI, \lambda)$ to $A$. Let $\lambda f^{-1}$ be the societal summary (also denoted by $s(f)$), which is the action distribution of all the players. That is, $s(f) (a_k)$ (or the $k$-th component of $s(f)$ when viewed as a vector in the simplex $\Delta$) is the proportion of players taking action $a_k$ for any $1 \le k \le K$.

Let $(\Omega, \cF, \bP)$ be a sample probability space that captures all the uncertainty associated with the individual randomization when all the players play mixed strategies. Throughout the paper, we assume that the player space $(I, \cI, \lambda)$ together with the sample space $(\Omega, \cF, \bP)$ allows a rich Fubini extension.\footnote{A Fubini extension extends the usual product probability space and retains the Fubini property, which is used in \cite{Sun2006} to address the issue of joint measurability and to prove the exact law of large numbers for a continuum of independent random variables. A Fubini extension is rich if it has a continuum of independent and identically distributed random variables with the uniform distribution on $[0, 1]$; see \citet[Section 5]{Sun2006}, \cite{SZ2009}, and \cite{Podczeck2010} for constructions of rich Fubini extensions. The formal definitions are stated in Definition~\ref{defn:fubini} below. Hereafter, any measurable function on $I \times \Omega$ is meant to be measurable with respect to a rich Fubini extension.} For each player, a mixed strategy $\xi$ is a measurable function from $\Omega$ to $A$. The induced action distribution of the strategy $\xi$ is denoted by $\bP \xi^{-1}$, i.e., $\bP \xi^{-1} (a_k) = \bP \bigl( \xi^{-1} (a_k) \bigr)$ is the probability of taking action $a_k$ under $\xi$ for any $1 \le k \le K$. A \emph{mixed strategy profile} is a measurable function $g \colon I \times \Omega \to A$ with \emph{essentially pairwise independent} random actions in the sense that for $\lambda$-almost all $i \in I$, the random actions $g(i, \cdot)$ and $g(j, \cdot)$ are independent for $\lambda$-almost all $j \in I$.\footnote{A mixed strategy profile captures the idea of independent randomizations of individual agents' mixed strategies explicitly. Given a mixed strategy profile $g$, $i \mapsto \bP g_i^{-1}$ defines a measurable function from $I$ to $\Delta$, which is also called a randomized strategy profile in Subsection \ref{sec:appen:randomized}. Namely, a randomized strategy profile indicates only the action distributions of individual agents without modeling the cross-agent independent randomizations explicitly.} Moreover, a mixed strategy profile $g$ is said to be \emph{symmetric} if for any two players $i$ and $i'$, $\bP g_i^{-1} = \bP g_{i'}^{-1}$ whenever $u_i = u_{i'}$.

A mixed strategy profile $g$ is said to be a \emph{Nash equilibrium} if for $\lambda$-almost all player $i \in I$,
$$ \int_\Omega u_i \bigl( g_i(\omega), s(g_\omega) \bigr) \dif \bP(\omega)  \ge  \int_\Omega u_i \bigl( \xi(\omega), s(g_\omega) \bigr) \dif \bP(\omega) \text{ for any mixed strategy } \xi, $$
where $s(g_\omega) = \lambda g_\omega^{-1}$ is the societal summary of $g_\omega$. The inequality above requires that almost every player's equilibrium strategy maximizes her expected payoff, which takes account of other players' strategies. A pure strategy profile $f \colon I \to A$ is said to be a \emph{Nash equilibrium} if the degenerated mixed strategy profile $g \colon I \times \Omega \to A$ with $g(i, \omega) = f(i)$ for all $(i, \omega) \in I \times \Omega$ is a Nash equilibrium.

\subsection{Admissibility and aggregate robustness}
\label{sec:mixPE:properties}

In this subsection, we rigorously define the desirable properties: admissibility and aggregate robustness. From now onwards, when there is no confusion, a strategy (profile) usually refers to a mixed strategy (profile).

We first introduce the property of admissibility. A strategy profile satisfying admissibility should preclude weakly dominated strategies. For a player $i \in I$, a strategy $\xi$, and a societal summary $\tau \in \Delta$, we abuse the notation by using $u_i(\xi, \tau)$ to denote the expected payoff $\int_\Omega u_i \bigl( \xi(\omega), \tau \bigr) \dif \bP (\omega)$.

\begin{defn}[Admissibility]
\label{defn:admissibility}
\rm
A pure strategy $a \in A$ for player $i$ is said to be a \emph{weakly dominated strategy} if there exists a strategy $\xi \colon \Omega \to A$ such that $u_i(a, \tau) \le u_i(\xi, \tau)$ for each $\tau \in \Delta$ and $u_i(a, \tau') < u_i(\xi, \tau')$ for some $\tau' \in \Delta$.

A strategy profile $g \colon I \times \Omega \to A$ is said to be \emph{admissible} if for $\lambda$-almost all $i \in I$, the induced distribution $\bP g_i^{-1}$ puts zero probability on any weakly dominated strategy.
\end{defn}

Next, we describe the property of aggregate robustness. Intuitively, a strategy profile is aggregate robust if almost all player's strategy is a best response to a sequence of full-support perturbations on the societal summary. Let $\cM(\Delta)$ be the space of Borel probability measures on the simplex $\Delta$ endowed with the topology of weak convergence of measures, and $\ehat{\cM}(\Delta)$ the subset of Borel probability measures with full support on $\Delta$. For any player $i \in I$, action $a \in A$, and $\ehat{\tau} \in \cM(\Delta)$, we abuse the notation by using $u_i (a, \ehat{\tau})$ to denote the expected payoff $\int_{\Delta} u_i(a, \tau') \dif\ehat{\tau}(\tau')$, when player $i$ takes action $a$ and views the societal summary as a probability distribution on $\Delta$. A \emph{perturbation function} $\varphi$ is a Borel measurable mapping from $\Delta$ to $\ehat{\cM}(\Delta)$.

\begin{defn}[Aggregate robustness]
\label{defn:agg-robust}
\rm
A strategy profile $g \colon I \times \Omega \to A$ is said to be \emph{aggregate robust} if for $\lambda$-almost all player $i \in I$, there exists a sequence of perturbation functions $\{\varphi^n\}_{n \in \bZ_+}$ such that
\begin{itemize}
\item for any societal summary $\tau \in \Delta$, $\varphi^n (\tau)$ weakly converges to the Dirac measure $\delta_\tau$ at $\tau$;
\item for each $n \in \bZ_+$,
$$ \int_\Omega u_i \Bigl( g_i(\omega), \varphi^n \bigl( s(g_\omega) \bigr) \Bigr) \dif \bP(\omega) \ge \int_\Omega u_i \Bigl( \xi(\omega), \varphi^n \bigl( s(g_\omega) \bigr) \Bigr) \dif \bP(\omega) \text{ for any strategy } \xi, $$
where $s(g_\omega) = \lambda (g_\omega)^{-1} \in \Delta$ is the aggregate action distribution of $g_\omega$.
\end{itemize}
\end{defn}

The aggregate robustness requires that there is a sequence of perturbation functions converging to the societal summary, and almost all player's original strategy is always a best response to those perturbed societal summaries. As noted in the proof of Theorem~\ref{thm:robust} below, our definition of aggregate robustness for large games is strong enough to imply admissibility.

\subsection{Robust perfect equilibrium: concept and results}
\label{sec:mixPE:robustPE}

In this subsection, we introduce the notion of robust perfect equilibrium. Such an equilibrium is shown to exist with the properties of admissibility, aggregate robustness, and ex post robust perfection.

We first state the notion of $\varepsilon$-robust perfect equilibrium.

\begin{defn}[$\varepsilon$-robust perfect equilibrium]
\label{defn:e-mixPE}
\rm
A strategy profile $g^\varepsilon \colon I \times \Omega \to A$ is said to be an \emph{$\varepsilon$-robust perfect equilibrium} if
\begin{enumerate}[label=(\roman*)]
\item for $\lambda$-almost all $i \in I$, $\bP (g_i^\varepsilon)^{-1}$ has a full support on $A$,
\item there exist a measurable set of players $I^\varepsilon \subseteq I$ and a perturbation function $\varphi^\varepsilon$ such that
    \begin{itemize}
    \item $\lambda(I^\varepsilon) > 1-\varepsilon$, and $\varphi^\varepsilon(\tau)$ puts $(1-\varepsilon)$-weight on $\delta_\tau$ for each $\tau \in \Delta$,
    \item for each player $i \in I^\varepsilon$, for any $a_m$ and $a_\ell$ in $A$,
        \begin{equation}
        \label{eq:e-mixPE}
        \int_\Omega u_i \Bigl( a_m, \varphi^\varepsilon \bigl( s(g_\omega^\varepsilon) \bigr) \Bigr) \dif \bP(\omega)
        < \int_\Omega u_i \Bigl( a_\ell, \varphi^\varepsilon \bigl( s(g_\omega^\varepsilon) \bigr) \Bigr) \dif \bP(\omega)  \Rightarrow  \bP (g_i^\varepsilon)^{-1}(a_m) \le \varepsilon.
        \end{equation}
    \end{itemize}
\end{enumerate}
\end{defn}

In Condition (i), $\bP (g_i^\varepsilon)^{-1}$ describes player $i$'s probabilities of choosing actions under her $\varepsilon$-equilibrium strategy $g_i^\varepsilon$. The requirement of full support captures the idea that player $i$ may ``tremble'' and play any one of her actions. In Condition (ii), we introduce $I^\varepsilon$ to denote the set of ``rational'' players, which is called a \emph{perturbed player space}. To guarantee a strategy profile to be an $\varepsilon$-robust perfect equilibrium, it requires that most players are rational, and each of them assigns a relatively small probability to non-best-response actions with respect to the perturbed societal summary.

Then we define the robust perfect equilibrium as the ``limit'' of a sequence of $\varepsilon$-robust perfect equilibria.

\begin{defn}[Robust perfect equilibrium]
\label{defn:mixPE}
\rm
A strategy profile $g \colon I \times \Omega \to A$ is said to be a \emph{robust perfect equilibrium} if there exist a sequence of strategy profiles $\{g^n\}_{n \in \bZ_+}$ and a sequence of positive constants $\{\varepsilon_n\}_{n \in \bZ_+}$ such that
\begin{enumerate}[label=(\roman*)]
\item each $g^n$ is an $\varepsilon_n$-robust perfect equilibrium associated with a perturbed player space $I^n$ and a perturbation function $\varphi^n$, and $\varepsilon_n \to 0$ as $n$ goes to infinity,
\item for $\lambda$-almost all $i \in I$, there exists a subsequence $\{g^{n_k}\}_{k=1}^{\infty}$ such that $i \in \cap_{k=1}^\infty I^{n_k}$ and $\lim\limits_{k \to \infty} \bP (g^{n_k}_i)^{-1} = \bP g_i^{-1}$,
\item $\lim\limits_{n \to \infty} \int_I \bP (g^n_i)^{-1} \dif\lambda(i) = \int_I \bP g_i^{-1} \dif\lambda(i)$.
\end{enumerate}
\end{defn}

Recall that each $I^n$ is a set of rational players for $g^n$. Condition (ii) simply means that almost all player's equilibrium strategy can be approximated in distribution by a sequence of her ``rational'' strategies. Condition (iii) says that the average action distributions in the perturbed sequence converge to the average equilibrium action distribution.

A pure strategy profile $f \colon I \to A$ is said to be a \emph{pure strategy robust perfect equilibrium} if the degenerated mixed strategy profile $g \colon I \times \Omega \to A$ with $g(i, \omega) = f(i)$ for all $(i, \omega) \in I \times \Omega$ is a robust perfect equilibrium.

Note that the notion of robust perfect equilibrium allows three types of perturbations: (i) perturbations on the player space---a small number of players are allowed to be ``irrational''; (ii) perturbations in players' strategies---individual players may make mistakes; (iii) perturbations in societal summaries---reported societal summaries may be inaccurate.

We are now ready to present our first main result.

\begin{thm}
\label{thm:robust}
Any robust perfect equilibrium is admissible and aggregate robust. Moreover, we have the following properties.
\begin{enumerate}[label=(\arabic*)]
\item Any robust perfect equilibrium $g$ is ex post robust perfect in the sense that for $\bP$-almost all $\omega \in \Omega$, $g_\omega \colon I \to A$ is a pure strategy robust perfect equilibrium.
\item Any robust perfect equilibrium is a Nash equilibrium.
\end{enumerate}
\end{thm}

Theorem \ref{thm:robust} (1) indicates that the concept of robust perfect equilibrium has the ex post preservation property, which establishes the relationship between a robust perfect equilibrium and its realized pure strategy profiles. Thus, Theorem~\ref{thm:robust} shows that a robust perfect equilibrium has all the desired properties as discussed earlier. The next question is whether such an equilibrium exists, especially in symmetric mixed strategies. The following theorem shows the existence of robust perfect equilibria in symmetric mixed strategies and in pure strategies.

\begin{thm}
\label{thm:exist}
The following existence results hold:
\begin{enumerate}[label=(\arabic*)]
\item Every large game has a symmetric mixed strategy robust perfect equilibrium.
\item Every large game has a pure strategy robust perfect equilibrium.
\end{enumerate}
\end{thm}

\section{Examples}
\label{sec:example}

In this section, we first revisit the three-path game in Section \ref{sec:intro}. It will be shown that the game has a unique symmetric robust perfect equilibrium and the traffic flow induced by any (possibly non-symmetric) robust perfect equilibria is also unique. Two additional examples are provided to clarify the concepts of admissibility, aggregate robustness, and robust perfect equilibrium. In particular, Subsection \ref{sec:example:ad-ro} presents an admissible Nash equilibrium in a large game, which is not aggregate robust. A further example of a large game is given in Subsection \ref{sec:example:non-robust} showing that a Nash equilibrium with the properties of admissibility and aggregate robustness may not be a robust perfect equilibrium.

\subsection{A revisit of the three-path game in Section \ref{sec:intro}}
\label{sec:example:revisit}

As an illustration of robust perfect equilibrium, we revisit the three-path game considered in Section \ref{sec:intro}. Let $\tau = \bigl( \tau(a), \tau(b), \tau(c) \bigr)$ be a traffic flow (i.e., a distribution of drivers' choices), where each $\tau(p)$ describes the proportion of drivers choosing path $p$. Based on the time cost of each path, we define the payoff function $u$ of drivers, which depends on the chosen path and the traffic flow:
$$ u(a, \tau) = -\tfrac{1}{2}, \ u(b, \tau) = - \max\{\tau(b), \tfrac{1}{2}\}, \ u(c, \tau) = - \tau(c) - \tfrac{1}{3}. $$
Clearly, this payoff function is a decreasing function.

\paragraph{Determining all Nash equilibrium traffic flows.} Now we try to determine all the traffic flows induced by (mixed strategy) Nash equilibria. Let $g$ be a strategy profile from $I \times \Omega$ to $A = \{a, b, c\}$, and $\tau = \bigl( \tau(a), \tau(b), \tau(c) \bigr) = \int_I \bP g_i^{-1} \dif\lambda$ the average action distribution induced by $g$. It follows from the exact law of large numbers as in \cite{Sun2006} (see Lemma~\ref{lem:ELLN} below) that the societal action distribution $\lambda g_\omega^{-1}$ is $\tau$ for $\bP$-almost all $\omega \in \Omega$.

We determine the distribution $\tau$ as follows. Firstly, suppose $\tau(b) > \frac{1}{2}$ and hence $u(b, \tau) < -\tfrac{1}{2} = u(a, \tau)$, which implies that those drivers choosing $b$ with positive probabilities will have incentives to deviate by shifting the probabilities from choosing $b$ to $a$. Next, suppose $\tau(c) > \frac{1}{6}$ and hence $u(c, \tau) < -\tfrac{1}{2} = u(a, \tau)$, which means each driver choosing $c$ with positive probability will be better off if she shifts the probability from choosing $c$ to $a$. Finally, suppose $\tau(c) < \frac{1}{6}$ and $u(c, \tau) > -\tfrac{1}{2} = u(a, \tau) \ge u(b, \tau)$, then all the drivers who choose $a$ or $b$ can shift their probabilities from choosing $a$ or $b$ to $c$ to gain higher expected payoffs. Hence, in order for $g$ to be a Nash equilibrium, we must have $\tau(c) = \frac{1}{6}$ and $\tau(b) \in [0, \frac{1}{2}]$.

When $\tau(b) \in [0, \frac{1}{2}]$, it is a Nash equilibrium that each driver chooses the strategy $\bigl( \frac{5}{6}-\tau(b), \tau(b), \frac{1}{6}\bigr)$, since each action leads to the same payoff $-\frac{1}{2}$ given the action distribution $\tau$.

Therefore, the traffic flows $\bigl( \frac{5}{6}-\tau(b), \tau(b), \frac{1}{6} \bigr)$ with $\tau(b) \in [0, \frac{1}{2}]$ form all the Nash equilibrium traffic flows.

\paragraph{Unique Nash equilibrium traffic flow with admissibility and aggregate robustness.} For each $\tau(b) \in (0, \frac{1}{2}]$, the Nash equilibrium $g$ with average action distribution $\tau = \bigl( \frac{5}{6}-\tau(b), \tau(b), \frac{1}{6} \bigr)$ is shown to violate the properties of admissibility and aggregate robustness.

First, for each $\tau(b) \in (0, \frac{1}{2}]$, some drivers choose Path $b$ with positive probabilities in equilibrium. However, Path $b$ is weakly dominated by Path $a$ for each driver. Thus, the property of admissibility is violated in such an equilibrium.

Next, for each $\tau(b) \in (0, \frac{1}{2}]$, let $\ehat{\tau} \in \ehat{\cM}(\Delta)$ be any full-support perturbation of $\tau$. Since Path $b$ is weakly dominated by Path $a$, we have that $u(a, \ehat{\tau}) = \int_\Delta u(a, \tau') \dif\ehat{\tau}(\tau') > \int_\Delta u(b, \tau') \dif\ehat{\tau}(\tau') = u(b, \ehat{\tau})$. Hence, the strategies taken by those drivers choosing $b$ with positive probabilities will not be best responses to any sequence of full-support perturbations. It implies that $g$ violates the property of aggregate robustness.

\paragraph{The uniqueness of symmetric robust perfect equilibrium.} Now we turn to the traffic flow $(\frac{5}{6}, 0, \frac{1}{6})$. We have seen in Section~\ref{sec:intro} that the strategy profile $g^0$ in which each driver uses the strategy $\sigma^0 = (\frac{5}{6}, 0, \frac{1}{6})$ satisfies admissibility and aggregate robustness. Here we strengthen this result by showing that $g^0 \colon I \times \Omega \to A$ is a robust perfect equilibrium. Given the uniqueness of Nash equilibrium traffic flow with admissibility and aggregate robustness shown above, it follows  that $g^0$ is the unique symmetric robust perfect equilibrium. Under the strategy profile $g^0$, the strategic uncertainty ensures that each driver $i$ chooses $a$ with probability $\frac{5}{6}$ and $c$ with probability $\frac{1}{6}$.

For each $\varepsilon \in (0, \frac{1}{3})$, we consider a strategy profile $g^\varepsilon$ such that for every driver $i \in I$,
$$ \bP (g^\varepsilon_i)^{-1} = (\tfrac{5}{6}-\varepsilon, \varepsilon, \tfrac{1}{6}). $$
By the exact law of large numbers, for $\bP$-almost all $\omega \in \Omega$, the induced societal summary $s(g_\omega^\varepsilon) = \lambda (g^\varepsilon_\omega)^{-1}$ is
$$ s(g_\omega^\varepsilon) = \lambda (g^\varepsilon_\omega)^{-1} = \int_I \bP (g^\varepsilon_i)^{-1} \dif\lambda(i) = (\tfrac{5}{6}-\varepsilon, \varepsilon, \tfrac{1}{6}). $$
We use $\tau^\varepsilon$ to denote the societal summary $(\tfrac{5}{6}-\varepsilon, \varepsilon, \tfrac{1}{6})$.

Let $\varphi \colon \Delta \to \ehat{\cM}(\Delta)$ be a perturbation function on societal summaries such that $\varphi(\tau) = (1-\varepsilon) \delta_\tau + \frac{\varepsilon}{2} \delta_{(1,0,0)} +\frac{\varepsilon}{2} \eta$, where $\delta_\tau$ is the Dirac measure at $\tau \in \Delta$ (especially, $\delta_{(1,0,0)}$ is the Dirac measure at $(1,0,0)$), and $\eta$ is the uniform probability distribution on $\Delta$. For simplicity, we write $\varphi(\tau)$ as $\ehat{\tau}$.

Given $g^\varepsilon$, each driver $i$'s expected payoff, when choosing path $p$, is
$$ \int_\Omega u \bigl( p, \ehat{s(g_\omega^\varepsilon)} \bigr) \dif \bP(\omega)  =  \int_\Omega u \bigl( p, \ehat{\tau^\varepsilon} \bigr) \dif \bP(\omega)  =  u \bigl( p, \ehat{\tau^\varepsilon} \bigr). $$
If driver $i$ chooses Path $a$, then the expected payoff is $-\tfrac{1}{2}$ since $u ( a, \tau)$ is the constant $-\tfrac{1}{2}$. Similarly, if driver $i$ chooses Path $c$, then the expected payoff is
\begin{align*}
u \bigl( c, \ehat{\tau^\varepsilon} \bigr)
    & = (1-\varepsilon) u \bigl( c, (\tfrac{5}{6}-\varepsilon, \varepsilon, \tfrac{1}{6}) \bigr) + \tfrac{\varepsilon}{2} u \bigl( c, (1, 0, 0) \bigr) + \tfrac{\varepsilon}{2} \int_\Delta \bigl( - \tau' (c) -  \tfrac{1}{3} \bigr) \dif \eta(\tau')  \\
    & = (1-\varepsilon) ( -\tfrac{1}{6} - \tfrac{1}{3} ) + \tfrac{\varepsilon}{2} ( -\tfrac{1}{3}) + \tfrac{\varepsilon}{2} ( -\tfrac{1}{3} - \tfrac{1}{3} ) = -\tfrac{1}{2};
\end{align*}
recall that $\int_\Delta \tau' (c) \dif \eta(\tau') = \tfrac{1}{3}$ by symmetry. Finally, if driver $i$ chooses Path $b$, then the expected payoff is
\begin{align*}
u \bigl( b, \ehat{\tau^\varepsilon} \bigr)
    & = (1-\varepsilon) u \bigl( b, (\tfrac{5}{6}-\varepsilon, \varepsilon, \tfrac{1}{6}) \bigr) + \tfrac{\varepsilon}{2} u \bigl( b, (1, 0, 0) \bigr) + \tfrac{\varepsilon}{2} \int_\Delta u \bigl( b, \tau' \bigr) \dif \eta(\tau')   \\
    & = (1-\varepsilon) (-\tfrac{1}{2}) + \tfrac{\varepsilon}{2} (-\tfrac{1}{2}) + \tfrac{\varepsilon}{2} \int_\Delta u \bigl( b, \tau' \bigr) \dif \eta(\tau')    \\
    & < (1-\varepsilon) (-\tfrac{1}{2}) + \tfrac{\varepsilon}{2} (-\tfrac{1}{2}) + \tfrac{\varepsilon}{2} \int_\Delta \bigl(-\tfrac{1}{2} \bigr) \dif \eta(\tau') = -\tfrac{1}{2}.
\end{align*}
The inequality follows from the fact that $u \bigl( b, \tau' \bigr) \le -\tfrac{1}{2}$ for any $\tau'$ and $u \bigl( b, \tau' \bigr) = - \tau'(b) < -\tfrac{1}{2}$ if $\tau'(b) > \tfrac{1}{2}$.

Therefore, $a$ and $c$ are the best responses for each driver $i$. Since $\bP (g^\varepsilon_i)^{-1}$ has $(1-\varepsilon)$-weight on the set of best responses for each driver $i \in I$, $g^\varepsilon$ is an $\varepsilon$-robust perfect equilibrium (with the perturbed driver space being $I$).

For each $n \in \bZ_+$, we define a strategy profile $g^n$ from $g^\varepsilon$ by letting $\varepsilon = \frac{1}{6n}$. Then, each $g^n$ is a $\frac{1}{6n}$-robust perfect equilibrium (with the perturbed driver space being $I$). Now, letting $n$ approach infinity, we have that $\lim\limits_{n \to \infty} \bP (g_i^n)^{-1} = \lim\limits_{n \to \infty} (\frac{5}{6}-\frac{1}{6n}, \frac{1}{6n}, \tfrac{1}{6}) = (\frac{5}{6}, 0, \tfrac{1}{6}) = \bP (g_i^0)^{-1}$ for each $i \in I$ and $\lim\limits_{n \to \infty} \int_I \bP (g_i^n)^{-1} \dif\lambda(i) = \int_I \bP (g_i^0)^{-1} \dif\lambda(i)$. That is, $g^0$ is a robust perfect equilibrium.

Since each robust perfect equilibrium has the property of ex post robust perfection, for $\bP$-almost all $\omega \in \Omega$, the realized pure strategy profile $g_\omega^0$ is also a robust perfect equilibrium. Moreover, by the exact law of large numbers, we conclude that in the pure strategy robust perfect equilibrium $g^0_\omega$, $\frac{5}{6}$ of the drivers choose Path $a$ (the broad highway) and $\frac{1}{6}$ of the drivers choose Path $c$ (the narrow shortcut).

\subsection{An admissible Nash equilibrium without aggregate robustness}
\label{sec:example:ad-ro}

In the modified Pigou's game and the three-path game in Section~\ref{sec:intro}, we obtain equilibria with the desired properties of admissibility and aggregate robustness by eliminating the weakly dominated strategy. As mentioned in the end of Subsection \ref{sec:mixPE:properties}, aggregate robustness  implies admissibility.
In this subsection, we present a large game possessing an admissible Nash equilibrium without aggregate robustness.

We modify the game in \citet[Section~2.1]{AHP2007} as follows. 
Agents choose simultaneously among three actions: they can fiercely attack (action $a$), or mildly attack (action $b$), or refrain from attacking (action $c$). For each player $i$, her payoff function is given by
$$ u(a, \tau) = \theta_1 \bigl( \tau(a) - \tfrac{1}{2} \bigr), \ u(b, \tau) = \theta_2 \bigl( \tau(b) - \tfrac{1}{2} \bigr), \ u(c, \tau) = 0, $$
where $\theta_1$ and $\theta_2$ exogenously parameterize the relative cost of an attack, with $1 > \theta_1 > \theta_2 > 0$.

Consider a symmetric strategy profile $g^0$ in which every player uses the same mixed strategy $\sigma^0 = (\tfrac{1}{2},\tfrac{1}{2},0)$. The law of large numbers implies that the induced action distribution $\tau^0 = (\tfrac{1}{2}, \tfrac{1}{2}, 0)$. Since $u(a, \tau^0) = u(b, \tau^0) = u(c, \tau^0) = 0$, $g^0$ is a Nash equilibrium. Moreover, $g^0$ is admissible as no action is weakly dominated. However, we will show that $g^0$ does not satisfy the aggregate robustness below.

We prove this argument by contradiction. Suppose that $g^0$ satisfies the aggregate robustness. For each player $i$, her strategy $\sigma^0$ is a best response with respect to a sequence of full-support perturbations $\{\ehat{\tau^n}\}_{n\in\bZ_+}$, where $\ehat{\tau^n}$ weakly converges to $\delta_{\tau^0}$. Therefore, both $a$ and $b$ should be best responses for each perturbation $\ehat{\tau^n}$:
$$ \int_\Delta u(a, \tau') \dif\ehat{\tau^n}(\tau') = \int_\Delta u(b, \tau') \dif\ehat{\tau^n}(\tau') \ge \int_\Delta u(c, \tau') \dif\ehat{\tau^n}(\tau') = 0. $$
Then we have
$$ \int_\Delta \theta_1 \bigl(\tau'(a) - \tfrac{1}{2}\bigr) \dif\ehat{\tau^n}(\tau') \ge 0,  \text{ or }  \int_\Delta \bigl(\tau'(a) - \tfrac{1}{2}\bigr) \dif\ehat{\tau^n}(\tau') \ge 0. $$
Thus,
$$ \int_{\{\tau' \mid \tau'(a) < \frac{1}{2}\}} \bigl(\tau'(a) - \tfrac{1}{2}\bigr) \dif\ehat{\tau^n}(\tau')  +  \int_{\{\tau' \mid \tau'(a) \ge \frac{1}{2}\}} \bigl(\tau'(a) - \tfrac{1}{2}\bigr) \dif\ehat{\tau^n}(\tau')  \ge  0. $$
Since $\ehat{\tau^n}$ has a full support on $\Delta$ and $\{\tau' \mid \tau'(a) < \frac{1}{2}\}$ is the union of disjoint sets $\{\tau' \mid \tau'(b) > \frac{1}{2}\}$ and $\{\tau' \mid \tau'(a) < \frac{1}{2}, \tau'(b) \le \frac{1}{2}\}$, we know that
\begin{align*}
	& \int_{\{\tau' \mid \tau'(a) < \frac{1}{2}\}} \bigl(\tfrac{1}{2} - \tau'(a)\bigr) \dif\ehat{\tau^n}(\tau')   \\
={}	& \int_{\{\tau' \mid \tau'(b) > \frac{1}{2}\}} \bigl(\tfrac{1}{2} - \tau'(a)\bigr) \dif\ehat{\tau^n}(\tau')  +  \int_{\{\tau' \mid \tau'(a) < \frac{1}{2}, \tau'(b) \le \frac{1}{2}\}} \bigl(\tfrac{1}{2} - \tau'(a)\bigr) \dif\ehat{\tau^n}(\tau')   \\
>{}	& \int_{\{\tau' \mid \tau'(b) > \frac{1}{2}\}} \bigl(\tfrac{1}{2} - \tau'(a)\bigr) \dif\ehat{\tau^n}(\tau').
\end{align*}
Therefore,
\begin{align*}
\int_{\{\tau' \mid \tau'(a) \ge \frac{1}{2}\}} \bigl(\tau'(a) - \tfrac{1}{2}\bigr) \dif\ehat{\tau^n}(\tau')
	& \ge \int_{\{\tau' \mid \tau'(a) < \frac{1}{2}\}} \bigl(\tfrac{1}{2} - \tau'(a)\bigr) \dif\ehat{\tau^n}(\tau')	\\
	& > \int_{\{\tau' \mid \tau'(b) > \frac{1}{2}\}} \bigl(\tfrac{1}{2} - \tau'(a)\bigr) \dif\ehat{\tau^n}(\tau')	\\
	& \ge \int_{\{\tau' \mid \tau'(b) > \frac{1}{2}\}} \Bigl(\tfrac{1}{2} - \bigl(1-\tau'(b)\bigr)\Bigr) \dif\ehat{\tau^n}(\tau')	\\
	& = \int_{\{\tau' \mid \tau'(b) > \frac{1}{2}\}} \bigl(\tau'(b)- \tfrac{1}{2} \bigr) \dif\ehat{\tau^n}(\tau')	\\
	& = \int_{\{\tau' \mid \tau'(b) \ge \frac{1}{2}\}} \bigl(\tau'(b)- \tfrac{1}{2} \bigr) \dif\ehat{\tau^n}(\tau'),
\end{align*}
where the third inequality is due to $\tau'(a) \le 1 - \tau'(b)$. So we conclude that
$$ \int_{\{\tau' \mid \tau'(a) \ge \frac{1}{2}\}} \bigl(\tau'(a) - \tfrac{1}{2}\bigr) \dif\ehat{\tau^n}(\tau') > \int_{\{\tau' \mid \tau'(b) \ge \frac{1}{2}\}} \bigl(\tau'(b)- \tfrac{1}{2} \bigr) \dif\ehat{\tau^n}(\tau'). $$

Since action $b$ is also a best response with respect to $\ehat{\tau^n}$, the similar arguments show that
$$ \int_{\{\tau' \mid \tau'(b) \ge \frac{1}{2}\}} \bigl(\tau'(b)- \tfrac{1}{2} \bigr) \dif\ehat{\tau^n}(\tau') > \int_{\{\tau' \mid \tau'(a) \ge \frac{1}{2}\}} \bigl(\tau'(a) - \tfrac{1}{2}\bigr) \dif\ehat{\tau^n}(\tau'), $$
which leads to a contradiction.

\subsection{A non-robust perfect Nash equilibrium with admissibility and aggregate robustness}
\label{sec:example:non-robust}

Theorem~\ref{thm:robust} indicates that a robust perfect equilibrium is a Nash equilibrium with the properties of admissibility and aggregate robustness. A natural question is whether the converse is true. The following example presents a Nash equilibrium with admissibility and aggregate robustness, which is not a robust perfect equilibrium.

Let $I_1$, $I_2$ and $I_3$ be disjoint sets in $\cI$ with union being $I$ and equal measure $\frac{1}{3}$, and $A = \{ a, b, c \}$ the common action space for all the players. The common payoff function is given as follows:
$$ u(a, \tau) = 0; $$
$$ u(b, \tau) = \begin{cases}
-1,				& \text{if } 0 \le \tau(b) < 0.2,	\\
10\tau(b)-3,	& \text{if } 0.2 \le \tau(b) < 0.3,	\\
0,				& \text{if } 0.3 \le \tau(b) < 0.8,	\\
5\tau(b)-4,		& \text{if } 0.8 \le \tau(b) \le 1;
\end{cases}
\qquad u(c, \tau) = \begin{cases}
-1,				& \text{if } 0 \le \tau(c) < 0.2,	\\
10\tau(c)-3,	& \text{if } 0.2 \le \tau(c) < 0.3,	\\
0,				& \text{if } 0.3 \le \tau(c) < 0.8,	\\
5\tau(c)-4,		& \text{if } 0.8 \le \tau(c) \le 1.
\end{cases} $$

\noindent Let $f$ be a pure strategy profile such that for each $i \in I$,
$$ f(i) = \begin{cases}
a,  & \text{if } i \in I_1,  \\
b,  & \text{if } i \in I_2,  \\
c,  & \text{if } i \in I_3.
\end{cases}$$

\begin{claim}
\label{claim:example:non-robust}
The pure strategy profile $f$ is a Nash equilibrium with the properties of admissibility and aggregate robustness. However, it is not a robust perfect equilibrium.
\end{claim}

\section{Applications: Congestion games and potential games}
\label{sec:application}

In a congestion game, the payoff of each player depends on the path she chooses and the proportion of players choosing the same or some overlapping paths. Congestion games are a special class of potential games. Both congestion games and potential games have been widely studied and found applications in many areas.\footnote{See \cite{Wardrop1952}, \cite{Rosenthal1973}, \cite{MS1996}, \cite{Milchtaich2000, Milchtaich2004}, \cite{AV2001}, \cite{RT2002}, \cite{AO2007}, \cite{MW2009}, \cite{CM2014}, \cite{FILRS2016}, and \cite{AMMO2018} for example.} In this section, we study robust perfect equilibrium in congestion games and potential games with many players. In particular, Subsection~\ref{sec:appli:congestion} shows the existence of pure strategy robust perfect equilibria in atomless congestion games, and also the uniqueness of symmetric robust perfect equilibrium when the cost functions are strictly increasing. An optimization method is proposed in Subsection~\ref{sec:appli:algorithm} for finding $\varepsilon$-robust perfect equilibria. It is shown in Subsection~\ref{sec:appli:congestion-NEPE} that an admissible Nash equilibrium must be robust perfect in any two-path congestion game. The final subsection considers large potential games.

\subsection{Atomless congestion games}
\label{sec:appli:congestion}

An atomless congestion game can be described by a network $N = (V, E)$, where $V$ represents the set of nodes and $E$ is the set of edges. Each path between an origin node and a destination node is a sequence of distinct edges. We focus on single origin-destination congestion games here. The player space is modeled by an atomless probability space $(I, \cI, \lambda)$ as in Subsection~\ref{sec:mixPE:game}. Let $\cP$ denote the set of all paths between the origin and the destination. So $\cP$ is the set of available choices/actions to each player. In this section, we slightly abuse the notation by using $\Delta$ to denote the set of probability distributions on $\cP$. Given an action profile, for each path $p \in \cP$, we use $\tau(p)$ to denote the mass of players on path $p$. Each edge $e \in E$ is associated with a cost function $C_e \bigl( \tau(e) \bigr)$, where
$$ \tau(e) = \sum_{\{p \in \cP \mid e \in p\}} \tau(p) $$
denotes the mass of players passing through the edge $e$ and $\{p \in \cP \mid e \in p\}$ denotes the set of all the paths that traverse the edge $e$. Each cost function $C_e$ is assumed to be nonnegative, continuous, and increasing. Thus, given an action distribution $\tau = \bigl( \tau(p) \bigr)_{p \in \cP}$, the cost function of path $p \in \cP$ is given by
$$ C_p(\tau)  =  \sum_{\{e \in E \mid e \in p\}} C_e \bigl( \tau(e) \bigr)  =  \sum_{\{e \in E \mid e \in p\}} C_e \biggl( \sum_{\{p' \in \cP \mid e \in p'\}} \tau(p') \biggr). $$
Note that each player's cost (and payoff) depends only on her own choice and the action distribution of all players, and hence a congestion game is a large game. Based on Theorem~\ref{thm:exist}, we have the following existence result whose proof is in Subsection~\ref{sec:appen:application}.

\begin{prop}
\label{prop:congestion-PE-exist}
Every atomless congestion game has a symmetric mixed strategy robust perfect equilibrium and a pure strategy robust perfect equilibrium. Moreover, if every cost function $C_e$ is strictly increasing, then there is a unique symmetric robust perfect equilibrium and the traffic flow induced by any (possibly non-symmetric) robust perfect equilibria is unique.
\end{prop}

When the cost functions are strictly increasing in a congestion game, \citet[Section~3.1.3, p.~3.13]{BMW1956} showed the uniqueness of Wardrop equilibrium. Proposition~\ref{prop:congestion-PE-exist} indicates that the unique Wardrop equilibrium satisfies not only the individual optimality condition as in the literature, but also the additional important properties of admissibility and aggregate robustness as considered in this paper.

Given a distribution $\tau$ of all players' actions, the social cost is given by
$$ C(\tau)  =  \sum_{e \in E} C_e \bigl( \tau(e) \bigr) \tau(e). $$
In congestion games, the price of anarchy measures the inefficiency of equilibria. One version of price of anarchy is defined as the ratio between the highest social cost in Nash equilibria and the minimum-possible social cost. \citet[Section~3.1.3]{BMW1956} showed that in any atomless congestion game, all Nash equilibria lead to the same social cost. Since a robust perfect equilibrium is a Nash equilibrium, the social cost in a robust perfect equilibrium is the same as the social cost in a Nash equilibrium. Therefore, the price of anarchy is equal to the ratio between the highest social cost in robust perfect equilibria and the minimum-possible social cost.

\subsection{Finding $\varepsilon$-robust perfect equilibria in congestion games}
\label{sec:appli:algorithm}

Although Proposition~\ref{prop:congestion-PE-exist} guarantees the existence of a robust perfect equilibrium in a congestion game, it does not seem easy to identify a robust perfect equilibrium based on its definition. Each robust perfect equilibrium is defined as the limit of a sequence of $\varepsilon$-robust perfect equilibria. So to obtain a robust perfect equilibrium, it suffices to find a sequence of $\varepsilon_n$-robust perfect equilibrium with $\lim\limits_{n \to \infty} \varepsilon_n = 0$. It is followed that an $\varepsilon$-robust perfect equilibrium exists via fixed-point arguments and related approaches. Below we provide a direct method to find an $\varepsilon$-robust perfect equilibrium by solving the corresponding optimization problem.

For any $\varepsilon \in (0, \tfrac{1}{|\cP|})$\footnote{We use $|\cP|$ to denote the cardinality of the set $\cP$.} and for any action distribution $\tau$, we define a perturbed cost function for each edge $e \in E$:
$$ C_e^{\varepsilon}\bigl(\tau(e)\bigr) = (1 - \varepsilon) C_e\bigl(\tau(e)\bigr) + \varepsilon \int_\Delta C_e\bigl(\tau'(e)\bigr) \dif \eta(\tau'), $$
where $\Delta$ is the $|\cP|$-dimensional unit simplex and $\eta$ is the uniform distribution on $\Delta$. Consider the following optimization problem:
\begin{equation}
\label{eq:optimization}
\tag{$*$}
\begin{aligned}
& \underset{\tau = \left( \tau(p) \right) \in \Delta}{\text{minimize}}	& & \sum_{e \in E} \int_0^{\tau(e)} C_e^{\varepsilon}(x) \dif x \\
& \text{subject to}								& & \sum\limits_{\{p \in \cP | e \in p\}} \tau(p) = \tau(e) \text{ for all } e \in E,   \\
&                                               & & \sum\limits_{p \in \cP} \tau(p) = 1 \text{ and } \tau(p) \ge \varepsilon \text{ for all } p \in \cP.
\end{aligned}
\end{equation}

\begin{prop}
\label{prop:algorithm}
Suppose that $\tau^\varepsilon = \bigl( \tau^\varepsilon(p) \bigr)$ is a solution to the optimization problem~\eqref{eq:optimization}. Let $g^\varepsilon$ be a strategy profile where each player uses the same (mixed) strategy and chooses path $p$ with probability $\tau^\varepsilon(p)$. Then, $g^\varepsilon$ is an $\varepsilon$-robust perfect equilibrium.
\end{prop}

\begin{rmk}
\label{rmk:algorithm}
\rm
Since $C_e^\varepsilon \bigl( \tau(e) \bigr)$ is an increasing function of $\tau(e)$, the objective function is convex, and hence the optimization problem is a convex problem. Thus, this problem can be solved by several well-established methods, including the simplex method, the penalty and augmented Lagrangian method, and interior-point method; see \cite{NW2006} for more details.

Pick a sequence $\{\varepsilon_n\}_{n \in \bZ_+}$ with $\varepsilon_n \to 0$. For each $\varepsilon_n$, we get a solution $\tau^n \in \Delta$ for the corresponding optimization problem. Since $\Delta$ is compact, the sequence $\{\tau^n\}_{n \in \bZ_+}$ has a limit point $\tau^*$. Then one can prove that every pure strategy profile with the action distribution $\tau^*$ is a robust perfect equilibrium; the details are given in Subsection~\ref{sec:appen:application}. Therefore, Proposition~\ref{prop:algorithm} provides an alternative proof of the existence of a pure strategy robust perfect equilibrium in congestion games.
\end{rmk}

Now we apply the optimization method above to identify a robust perfect equilibrium for the three-path game in Section~\ref{sec:intro}. Given a sequence $\{\varepsilon_n\}_{n \in \bZ_+}$ in the interval $(0, \tfrac{1}{10})$ with $\varepsilon_n \to 0$, the perturbed cost functions are:
\begin{align*}
C_a^{\varepsilon_n}\bigl( \tau(a) \bigr)	& = (1 - \varepsilon_n) C_a\bigl( \tau(a) \bigr) + \varepsilon_n \int_\Delta C_a \bigl( \tau'(a) \bigr) \dif \eta(\tau') = \tfrac{1}{2},	\\
C_b^{\varepsilon_n}\bigl( \tau(b) \bigr)	& = (1 - \varepsilon_n) C_c\bigl( \tau(b) \bigr) + \varepsilon_n \int_\Delta C_b \bigl( \tau'(b) \bigr) \dif \eta(\tau') > (1 - \varepsilon_n)\tfrac{1}{2} + \tfrac{\varepsilon_n}{2} = \tfrac{1}{2}, \\
C_c^{\varepsilon_n}\bigl( \tau(c) \bigr)	& = (1 - \varepsilon_n) C_c\bigl( \tau(c) \bigr) + \varepsilon_n \int_\Delta C_c \bigl( \tau'(c) \bigr) \dif \eta(\tau') \\ &= (1 - \varepsilon_n)\bigl(\tau(c) + \tfrac{1}{3}\bigr) + \varepsilon_n \int_\Delta \bigl(\tau'(c)+ \tfrac{1}{3}\bigr) \dif \eta(\tau') = (1-\varepsilon_n)\tau(c) + \tfrac{1+\varepsilon_n}{3}.
\end{align*}

We solve the following optimization problem:
$$\begin{aligned}
& \underset{\tau \in \Delta}{\text{minimize}}	& & \int_0^{\tau(a)} \tfrac{1}{2} \dif x + \int_0^{\tau(b)} C_b^{\varepsilon_n}( x ) \dif x + \int_0^{\tau(c)} \Bigl[ (1-\varepsilon_n)x + \tfrac{1+\varepsilon_n}{3} \Bigr] \dif x \\
& \text{subject to}								& & \tau(a) + \tau(b) + \tau(c) = 1, \ \tau(a) \ge \varepsilon_n, \ \tau(b) \ge \varepsilon_n, \ \tau(c) \ge \varepsilon_n.
\end{aligned}$$
Suppose that $\tau^n = \bigl( \tau^n(a), \tau^n(b), \tau^n(c) \bigr)$ is a solution to this optimization problem. We claim that $\tau^n(b) = \varepsilon_n$. Otherwise, we assume $\tau^n(b) > \varepsilon_n$. Now consider another distribution $\bigl(\tau^n(a) + \tau^n(b) -\varepsilon_n, \varepsilon_n, \tau^n(c)\bigr)$. Since $C_b^{\varepsilon_n}(x) > \frac{1}{2} = C_a^{\varepsilon_n}(x)$, this distribution further reduces the total cost, which leads to a contradiction. Thus, the optimization problem can be simplified to:
$$\begin{aligned}
& \underset{\tau \in \Delta}{\text{minimize}}	& & \int_0^{\tau(a)} \tfrac{1}{2} \dif x + \int_0^{\tau(c)} \Bigl[ (1-\varepsilon_n)x + \tfrac{1+\varepsilon_n}{3} \Bigr] \dif x \\
& \text{subject to}								& & \tau(a) + \tau(c) = 1 - \varepsilon_n, \ \tau(a) \ge \varepsilon_n, \ \tau(c) \ge \varepsilon_n.
\end{aligned}$$
By simple calculation, the objective function becomes
\begin{align*}
\tfrac{\tau(a)}{2} + \tfrac{1-\varepsilon_n}{2} \tau(c)^2 + \tfrac{1+\varepsilon_n}{3}\tau(c)	& = \tfrac{1-\varepsilon_n-\tau(c)}{2} + \tfrac{1-\varepsilon_n}{2} \tau(c)^2 + \tfrac{1+\varepsilon_n}{3}\tau(c)	\\
& = \tfrac{1-\varepsilon_n}{2} \tau(c)^2 - \tfrac{1-2\varepsilon_n}{6}\tau(c) + \tfrac{1-\varepsilon_n}{2}.
\end{align*}
It is easy to see that the objective function attains the minimum at $\tau(c) = \tfrac{1-2\varepsilon_n}{6-6\varepsilon_n} \in [\varepsilon_n, 1-2\varepsilon_n]$. Thus, the unique solution is $\tau^n = \bigl( \tau^n(a), \tau^n(b), \tau^n(c) \bigr) = \bigl(\tfrac{5-10\varepsilon_n+6\varepsilon_n^2}{6-6\varepsilon_n}, \varepsilon_n, \tfrac{1-2\varepsilon_n}{6-6\varepsilon_n}\bigr)$, which is induced by an $\varepsilon_n$-robust perfect equilibrium. Letting $\varepsilon_n \to 0$, we conclude that $\tau^* = (\tfrac{5}{6}, 0, \tfrac{1}{6})$ is an action distribution of a robust perfect equilibrium, which coincides with the result in Subsection~\ref{sec:example:revisit}.

\subsection{Robust perfect equilibria and admissible Nash equilibria in two-path congestion games}
\label{sec:appli:congestion-NEPE}

The example in Subsection~\ref{sec:example:non-robust} shows that the concept of robust perfect equilibrium is stronger than the notion of admissible Nash equilibrium even for a large game with three actions. Nevertheless, we show in the following proposition that every admissible Nash equilibrium in a two-path congestion game is always a robust perfect equilibrium.

\begin{prop}
\label{prop:congestion-NEPE}
In any atomless congestion game with two paths, every admissible Nash equilibrium is a robust perfect equilibrium.
\end{prop}

The proof is collected in Subsection~\ref{sec:appen:application}. This result implies that in several well-known examples, each admissible Nash equilibrium is a robust perfect equilibrium. In particular, for the Pigou's example illustrated in Figure~\ref{fig:Pigou}, the unique admissible Nash equilibrium $f(i) \equiv a$ is robust perfect, and for Braess's paradox in Figure~\ref{fig:Braess}, each admissible Nash equilibrium, where half of drivers choose the upper path and the rest choose the lower path, is also robust perfect.

\begin{figure}[htb!]
\centering
\subfigure[Pigou's example]{\label{fig:Pigou}\raisebox{11.5pt}[0pt][0pt]{\makebox[0.35\textwidth][c]{\includegraphics[scale=0.65]{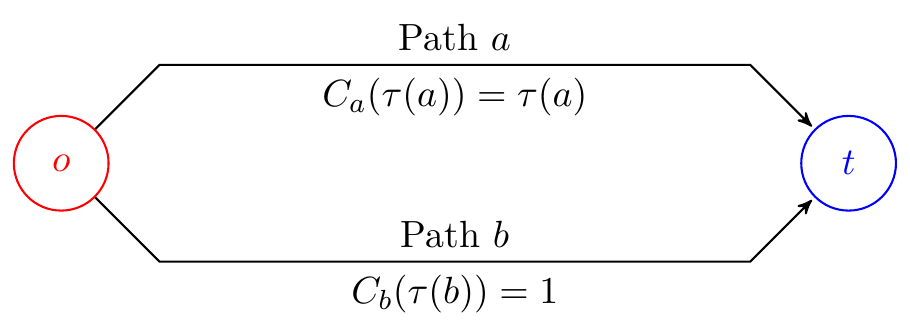}}}}
\qquad
\subfigure[Braess's paradox]{\label{fig:Braess}\includegraphics[width=0.45\textwidth]{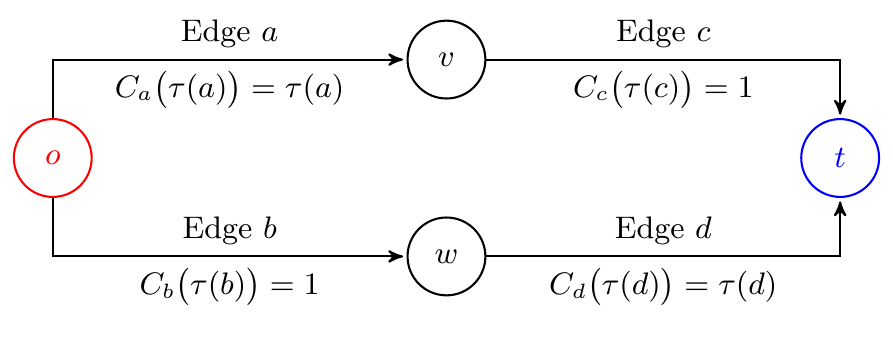}}
\caption{Pigou's example and Braess's paradox}
\label{fig:Pigou-Braess}
\end{figure}
\FloatBarrier

\subsection{Large potential games}
\label{sec:appli:potential}

We end this section by a brief discussion on large potential games. A strategic game with finite players is said to be a potential game if the incentives of all players to change their strategies can be expressed by one global function, which is called the potential function. The potential function is a useful tool to analyze equilibrium properties of games, since the incentives of all players are mapped into one function, and the set of pure strategy Nash equilibria can be found by locating the local optima of the potential function. For finite-player potential games, \cite{CM2014} consider the equilibrium refinement and show the existence of pure strategy trembling-hand perfect equilibria. In the following, we study large potential games and robust perfect equilibria therein.

\begin{defn}[Large potential game]
\label{defn:potential}
\rm
Given a large game $G$, a function $P \colon A \times \Delta \to \bR$ is a \emph{potential} for $G$ if for $\lambda$-almost all player $i \in I$, for all actions $a_m$ and $a_\ell$ in A, and for every $\tau \in \Delta$,
$$ u_i(a_m, \tau) - u_i(a_\ell, \tau)  =  P(a_m, \tau) - P(a_\ell, \tau). $$
A large game is a \emph{large potential game} if it admits a potential.
\end{defn}

It is easy to see that atomless congestion games are a special case of large potential games. Similarly, as an application of Theorem~\ref{thm:exist}, we can also establish the existence of pure strategy robust perfect equilibria in large potential games.

\begin{prop}
\label{prop:potential}
Each large potential game possesses a pure strategy robust perfect equilibrium.
\end{prop}

\section{Technical preparation and proofs}
\label{sec:appen}

\subsection{Some technical background}

For the convenience of readers, here we state the definition of a rich Fubini extension as in \cite{Sun2006}. Let probability spaces $(I, \cI, \lambda)$ and $(\Omega, \cF, \bP)$ be the index and sample spaces, respectively. Let $(I \times \Omega, \cI \otimes \cF, \lambda \otimes \bP)$ be the usual product probability space. Given any function $g$ on $I \times \Omega$ and any $(i, \omega) \in I \times \Omega$, let $g_i$ be the function $g(i, \cdot)$ on $\Omega$ and $g_\omega$ the function $g(\cdot, \omega)$ on $I$.

\begin{defn}[Rich Fubini extension]
\label{defn:fubini}
\rm
A probability space $(I \times \Omega, \cW, \bQ)$ extending the usual product probability space $(I \times \Omega, \cI \otimes \cF, \lambda \otimes \bP)$ is said to be a \emph{Fubini extension} of $(I \times \Omega, \cI \otimes \cF, \lambda \otimes \bP)$ if for any real-valued $\bQ$-integrable function $g$ on $(I \times \Omega, \cW)$,
\begin{enumerate}[label=(\roman*)]
\item the two functions $g_i$ and $g_\omega$ are integrable, respectively, on $(\Omega, \cF, \bP)$ for $\lambda$-almost all $i \in I$, and on $(I, \cI, \lambda)$ for $\bP$-almost all $\omega \in \Omega$;
\item $\int_\Omega g_i \dif\bP$ and $\int_I g_\omega \dif\lambda$ are integrable, respectively, on $(I, \cI, \lambda)$ and $(\Omega, \cF, \bP)$, with $\int_{I \times \Omega} g \dif\bQ = \int_I \bigl( \int_\Omega g_i \dif\bP \bigr) \dif\lambda = \int_\Omega \bigl( \int_I g_\omega \dif\lambda \bigr) \dif\bP$.
\end{enumerate}
To reflect the fact that the probability space $(I \times \Omega, \cW, \bQ)$ has $(I, \cI, \lambda)$ and $(\Omega, \cF, \bP)$ as its marginal spaces, as required by the Fubini property, it will be denoted by $(I \times \Omega, \cI \boxtimes \cF, \lambda \boxtimes \bP)$.

A Fubini extension is \emph{rich} if there is an $\cI \boxtimes \cF$-measurable process $g$ from $I \times \Omega$ to $[0, 1]$ such that the random variables $g_i(\cdot) = g(i, \cdot)$ have the uniform distribution on the unit interval $[0, 1]$, and are essentially pairwise independent in the sense that for $\lambda$-almost all $i \in I$, $g_i$ and $g_j$ are independent for $\lambda$-almost all $j \in I$.
\end{defn}

We state two lemmas for rich Fubini extensions, which will be frequently used in the proofs. The first, which is on the universality property of a rich Fubini extension, is taken from \citet[Proposition 5.3]{Sun2006}. It says that one can construct processes on a rich Fubini extension with essentially pairwise independent random variables that take any given variety of distributions.

\begin{lem}[Universality]
\label{lem:universality}
Let $(I \times \Omega, \cI \boxtimes \cF, \lambda \boxtimes \bP)$ be a rich Fubini extension, $X$ a Polish space (i.e., a complete separable metric space), and $h$ a measurable mapping from $(I, \cI, \lambda)$ to the space $\cM(X)$ of all Borel probability measures on $X$ endowed with the topology of weak convergence of measures. Then, there exists an $\cI \boxtimes \cF$-measurable process $g \colon I \times \Omega \to X$ such that the process $g$ is essentially pairwise independent and $h(i)$ is the induced distribution of $g_i$, for $\lambda$-almost all $i \in I$.
\end{lem}

The following lemma states a version of the exact law of large games (ELLN) in the framework of the Fubini extension, which is taken from \citet[Corollary~2.9]{Sun2006}.

\begin{lem}[Exact law of large numbers or ELLN]
\label{lem:ELLN}
Let $(I \times \Omega, \cI \boxtimes \cF, \lambda \boxtimes \bP)$ be a Fubini extension, $X$ a Polish space, and $g$ a process from $(I \times \Omega, \cI \boxtimes \cF, \lambda \boxtimes \bP)$ to $X$. If $g$ is essentially pairwise independent, then the sample distribution $\lambda g_\omega^{-1}$ is the same as the distribution $(\lambda \boxtimes \bP) g^{-1}$ for $\bP$-almost all $\omega \in \Omega$.
\end{lem}

\subsection{Robust perfect equilibria in randomized strategies}
\label{sec:appen:randomized}

In this subsection, we consider randomized strategy robust perfect equilibrium, which is a counterpart of mixed strategy robust perfect equilibrium. The equivalence result in Proposition \ref{prop:equivalence} below allows us to use a randomized strategy robust perfect equilibrium as an intermediate for obtaining a mixed strategy robust perfect equilibrium in the proof of Theorem \ref{thm:exist}.

A randomized strategy for a player is a probability distribution $\mu \in \Delta$. A \emph{randomized strategy profile} $h$ is a measurable function from $I$ to $\Delta$. Note that every pure strategy profile $f$ naturally corresponds to the randomized strategy profile $h^f$ where $h^f(i) = \delta_{f(i)}$ for each $i \in I$. Given a randomized strategy profile $h$, we model the societal summary $s(h)$ as the average action distribution of all the players $\int_I h(i) \dif\lambda(i) \in \Delta$. Also note that, when $f$ is a pure strategy profile, $\int_I f(i) \dif\lambda(i)$ reduces to $\lambda f^{-1}$, which is the societal action distribution induced by $f$.

\begin{defn}[Randomized strategy Nash equilibrium]
\label{defn:randNE}
\rm
A randomized strategy profile $h$ is said to be a \emph{randomized strategy Nash equilibrium} if for $\lambda$-almost all $i \in I$,
$$ \int_A u_i \bigl( a, s(h) \bigr) h(i; \dif a)  \ge  \int_A u_i \bigl( a, s(h) \bigr) \dif\mu(a) \text{ for all } \mu \in \Delta. $$
\end{defn}

A randomized strategy profile $h \colon I \to \Delta$ is said to \emph{have a full support} on $A$ if for $\lambda$-almost all $i \in I$, the strategy $h(i)$ assigns strictly positive probability to each action $a \in A$.

Similar to robust perfect equilibrium in mixed strategies, we also consider perturbations of societal summaries. Recall that a perturbation function $\varphi$ is a measurable mapping from $\Delta$ to $\ehat{\cM}(\Delta)$. Given a perturbed societal summary $\ehat{\tau} \in \ehat{\cM}(\Delta)$, we abuse the notion by using $u_i(a, \ehat{\tau})$ to denote player $i$'s expected payoff $\int_{\Delta} u_i(a, \tau') \dif\ehat{\tau}(\tau')$.

\begin{defn}[$\varepsilon$-robust perfect equilibrium in randomized strategies]
\label{defn:e-rPE}
\rm
A full-support randomized strategy profile $h^\varepsilon \colon I \to \Delta$ is said to be an \emph{$\varepsilon$-robust perfect equilibrium in randomized strategies} if there exist a measurable set of players $I^\varepsilon \subseteq I$ and a perturbation function $\varphi^\varepsilon$ such that
\begin{itemize}
\item $\lambda(I^\varepsilon) > 1-\varepsilon$ and $\varphi^\varepsilon(\tau)$ puts at least $(1-\varepsilon)$-weight on $\delta_\tau$ for each $\tau \in \Delta$,
\item for all $i \in I^\varepsilon$, for all $a_m$ and $a_\ell$ in $A$,
    \begin{equation}
    \label{eq:e-rPE}
    u_i \Bigl( a_m, \varphi^\varepsilon \bigl( s(h^\varepsilon) \bigr) \Bigr)  <  u_i \Bigl( a_\ell, \varphi^\varepsilon \bigl( s(h^\varepsilon) \bigr) \Bigr)  \Rightarrow  h^\varepsilon(i; a_m) \le \varepsilon.
    \end{equation}
\end{itemize}
\end{defn}

Here $I^\varepsilon$ also denotes the set of rational players. To guarantee a full-support randomized strategy profile to be an $\varepsilon$-robust perfect equilibrium, it is required that most players are rational---each rational player $i \in I^\varepsilon$ should assign relatively small probability to any non-best-response action with respect to a perturbation of the equilibrium societal summary.

Now we define robust perfect equilibrium in randomized strategies.

\begin{defn}[Robust perfect equilibrium in randomized strategies]
\label{defn:rPE}
\rm
A randomized strategy profile $h \colon I \to \Delta$ is said to be a \emph{randomized strategy robust perfect equilibrium} if there exist a sequence of randomized strategy profiles $\{h^n\}_{n \in \bZ_+}$ and a sequence of positive constants $\{\varepsilon_n\}_{n \in \bZ_+}$ such that
\begin{enumerate}[label=(\roman*)]
\item each $h^n$ is an $\varepsilon_n$-robust perfect equilibrium associated with the perturbed player space $I^n$ and a perturbation function $\varphi^n$, and $\varepsilon_n \to 0$ as $n$ goes to infinity,
\item for $\lambda$-almost all $i \in I$, there exists a subsequence $\{h^{n_k}\}_{k=1}^\infty$ such that $i \in \cap_{k=1}^\infty I^{n_k}$ and $\lim\limits_{k \to \infty} h^{n_k}(i) = h(i)$,
\item $\lim\limits_{n \to \infty} \int_I h^n(i) \dif\lambda(i) = \int_I h(i) \dif\lambda(i)$.
\end{enumerate}
\end{defn}

We present an equivalence result between a randomized strategy robust perfect equilibrium and a mixed strategy robust perfect equilibrium in the following proposition. Its proof is given in Subsection~\ref{sec:appen:equivalence}.

\begin{prop}
\label{prop:equivalence}
For a large game $G$, we have the following results.
\begin{enumerate}[label=(\arabic*)]
\item For any mixed strategy robust perfect equilibrium $g \colon I \times \Omega \to A$, the randomized strategy profile $h$ as defined by $h(i) = \bP g_i^{-1}$ for $i \in I$ is a robust perfect equilibrium in randomized strategy.
\item For any randomized strategy robust perfect equilibrium $h \colon I \to \Delta$, there is a mixed strategy robust perfect equilibrium $g$ with $\bP g_i^{-1} = h(i)$ for any $i \in I$.
\end{enumerate}
\end{prop}

When $h$ is a pure strategy profile, Proposition~\ref{prop:equivalence} implies that Definition~\ref{defn:rPE} coincides with the definition of pure strategy robust perfect equilibrium below Definition~\ref{defn:mixPE}.

The following remark discusses some relevant papers on various equilibrium concepts.

\begin{rmk}
\label{rmk:equivalence}
\rm
An equivalence result between a randomized strategy Nash equilibrium and a mixed strategy Nash equilibrium in large games was shown in \cite{KRSY2015}. Ex post Nash equilibrium in large games was considered in \citet[Section 11]{KS2002} and \cite{KRSY2015} while approximate ex post Nash equilibrium was studied in \cite{Kalai2004}.

For a large game with finite actions and general continuous payoffs, \cite{Schmeidler1973} and \cite{Rath1994,Rath1998} proved the existence of pure strategy Nash equilibria and pure strategy perfect equilibria, respectively. Unlike our concept of robust perfect equilibrium with three types of perturbations on the agent space, societal summaries, and actions by individual players, \cite{Rath1994,Rath1998} considered only perturbations on the actions by individual players. As noted in \citet[Section 5]{Rath1998}, such a pure strategy perfect equilibrium is not admissible in general. \cite{SZ2020} showed the existence of some version of randomized strategy perfect equilibrium without perturbations on the agent space. However, the definition of a pure strategy perfect equilibrium as in \citet[Definition 4]{SZ2020} is inconsistent with their definition of a randomized strategy perfect equilibrium when pure strategies are viewed as degenerate randomized strategies.\footnote{For a randomized strategy perfect equilibrium $g$ as in \citet[Definition 4]{SZ2020}, Part (2) of the definition requires that for almost all player $i$, $g(i)$ is a limit point of $g^n(i)$, where $\{g^n\}_{n \in \bZ_+}$ is a sequence $\epsilon_n$-perfect equilibrium in their sense. However, a pure strategy perfect equilibrium $g$ as in \citet[Definition 4]{SZ2020} does not require the Dirac measure $\delta_{g(i)}$ to be a limit point of $g^n(i)$; the pure strategy $g(i)$ over there was only required to be in the support of a limit point of $g^n(i)$.} None of those papers considered any version of mixed strategy perfect equilibrium, or aggregate robustness, or ex post robust perfection.
\end{rmk}

\subsection{Proof of Proposition~\ref{prop:equivalence}}
\label{sec:appen:equivalence}

The proof is divided into two parts.

\paragraph{Part~1.} Let $g$ be a mixed strategy robust perfect equilibrium and $h$ a randomized strategy profile such that $h(i) = \bP g_i^{-1}$ for each $i \in I$. We want to show that $h$ is a randomized strategy robust perfect equilibrium.

By Definition~\ref{defn:mixPE}, there exist a sequence of mixed strategy profiles $\{g^n\}_{n \in \bZ_+}$ and a sequence of positive constants $\{\varepsilon_n\}_{n \in \bZ_+}$ such that
\begin{enumerate}[label=(\roman*)]
\item each $g^n$ is an $\varepsilon_n$-robust perfect equilibrium associated with the perturbed player space $I^n$ and perturbation function $\varphi^n$, and $\varepsilon_n \to 0$ as $n$ goes to infinity,
\item for $\lambda$-almost all $i \in I$, there exists a subsequence $\{g^{n_k}\}_{k=1}^{\infty}$ such that $i \in \cap_{k=1}^\infty I^{n_k}$ and $\lim\limits_{k \to \infty} \bP (g^{n_k}_i)^{-1} = \bP g_i^{-1}$,
\item $\lim\limits_{n \to \infty} \int_I \bP (g^n_i)^{-1} \dif\lambda(i) = \int_I \bP g_i^{-1} \dif\lambda(i)$.
\end{enumerate}

For each $n \in \bZ_+$, we define a randomized strategy profile $h^n$ such that $h^n(i) = \bP (g_i^n)^{-1}$ for each $i \in I$. We claim that each $h^n$ is a randomized strategy $\varepsilon_n$-robust perfect equilibrium associated with the perturbed player space $I^n$ and perturbation function $\varphi^n$. By ELLN in Lemma~\ref{lem:ELLN}, there exists $\Omega^0 \subseteq \Omega$ with $\bP(\Omega^0) = 1$ such that for each $\omega \in \Omega^0$,
$$ s(h^n) = \int_I h^n(i) \dif\lambda(i) = \int_I \bP (g_i^n)^{-1} \dif\lambda(i) \xlongequal{\text{ELLN}} \lambda (g_\omega^n)^{-1} = s(g_\omega^n). $$

Consider two different actions $a_m$ and $a_\ell$ that satisfy $u_i \bigl( a_m, \varphi^n(s(h^n)) \bigr) < u_i \bigl( a_\ell, \varphi^n(s(h^n)) \bigr)$. Then we have that
$$ \int_\Omega u_i \bigl( a_m, \varphi^n(s(g_\omega^n)) \bigr) \dif \bP(\omega) < \int_\Omega u_i \bigl( a_\ell, \varphi^n(s(g_\omega^n)) \bigr) \dif \bP(\omega). $$
Since $g^n$ is an $\varepsilon_n$-robust perfect equilibrium, we have that for each $i \in I^n$,
$$h^n(i; a_m) = \bP (g_i^n)^{-1} (a_m) \le \varepsilon_n.$$
Moreover, since each $\bP(g^n_i)^{-1}$ is a full-support probability measure, so is $h^n$. Therefore, $h^n$ is an $\varepsilon_n$-robust perfect equilibrium with the perturbed player space $I^n$ and perturbation function $\varphi^n$.

Furthermore, for $\lambda$-almost all $i \in I$, recall that $\{g^{n_k}\}_{k=1}^\infty$ is the subsequence in Condition (ii). We consider the corresponding subsequence $\{h^{n_k}\}_{k=1}^\infty$. Then, $i \in \cap_{k=1}^\infty I^{n_k}$ and $\lim\limits_{k \to \infty} h^{n_k}(i) =\lim\limits_{k \to \infty} \bP (g^{n_k}_i)^{-1} = \bP g_i^{-1} = h(i)$. We also have $\lim\limits_{n \to \infty} \int_I h^n(i) \dif\lambda(i) = \lim\limits_{n \to \infty} \int_I \bP (g^n_i)^{-1} \dif\lambda(i) = \int_I \bP g_i^{-1} = \int_I h(i) \dif\lambda(i)$.

Therefore, $h$ is a randomized strategy robust perfect equilibrium.

\paragraph{Part~2.} Let $h$ be a randomized strategy robust perfect equilibrium. By Definition~\ref{defn:rPE}, there exist a sequence of randomized strategy profiles $\{h^n\}_{n \in \bZ_+}$ and a sequence of positive constants $\{\varepsilon_n\}_{n \in \bZ_+}$ such that
\begin{enumerate}[label=(\roman*)]
\item each $h^n$ is an $\varepsilon_n$-robust perfect equilibrium associated with the perturbed player space $I^n$ and perturbation function $\varphi^n$, and $\varepsilon_n \to 0$ as $n$ goes to infinity,
\item for $\lambda$-almost all $i \in I$, there exists a subsequence $\{h^{n_k}\}_{k=1}^\infty$ such that $i \in \cap_{k=1}^\infty I^{n_k}$ and $\lim\limits_{k \to \infty} h^{n_k}(i) = h(i)$,
\item $\lim\limits_{n \to \infty} \int_I h^n(i) \dif\lambda(i) = \int_I h(i) \dif\lambda(i)$.
\end{enumerate}

Note that $(I \times \Omega, \cI \boxtimes \cF, \lambda \boxtimes \bP)$ is a rich Fubini extension and $h$ is a measurable function from $I$ to $\Delta$. By Lemma~\ref{lem:universality}, there exists an $\cI \boxtimes \cF$-measurable function $f$ from $I \times \Omega$ to $A$ and a set $I' \in {\cal I}$ with $\lambda(I')=1$ such that $f$ is essentially pairwise independent and for any $i \in I'$, $\bP f_i^{-1} = h(i)$.
For any $i \in I \backslash I'$, let $f'_i: \Omega \to A$ be a random variable with distribution $h(i)$.
For any $(i, \omega) \in I \times \Omega$, let
$$g(i, \omega)=
\begin{cases}
f(i, \omega) & \text{ if }	i \in I',\\
f'(i, \omega) & \text{ if } i \notin I'.
\end{cases}
$$ It is clear that $g$ is also $\cI \boxtimes \cF$-measurable and essentially pairwise independent, and for any $i \in I$, $\bP g_i^{-1} = h(i)$.
Similarly, we obtain a sequence of $\cI \boxtimes \cF$-measurable functions $\{g^n\}_{n \in \bZ_+}$ such that for each $n \in \bZ_+$, $g^n \colon I \times \Omega \to A$ is essentially pairwise independent, and for every $i \in I$, $\bP (g_i^n)^{-1} = h^n(i)$.

By the similar arguments in Step 1, we can apply the ELLN (see Lemma~\ref{lem:ELLN}) to show that each $g^n$ is a mixed strategy $\varepsilon_n$-robust perfect equilibrium. The remaining proof is the same as Part~1 and hence $g$ is a mixed strategy robust perfect equilibrium.

\subsection{Proof of Theorem~\ref{thm:robust}}
\label{sec:appen:robust}

Let $g \colon I \times \Omega \to A$ be a mixed strategy robust perfect equilibrium. Then there exist a sequence of mixed strategy profiles $\{g^n\}_{n \in \bZ_+}$ and a sequence of positive constants $\{\varepsilon_n\}_{n \in \bZ_+}$ such that
\begin{enumerate}[label=(\roman*)]
\item each $g^n$ is an $\varepsilon_n$-robust perfect equilibrium associated with the perturbed player space $I^n$ and perturbation function $\varphi^n$, and $\varepsilon_n \to 0$ as $n$ goes to infinity,
\item there exists $I^0 \in {\cal I}$ such that $\lambda(I^0)=1$ and for any $i \in I^0$, there exists a subsequence $\{g^{n_k}\}_{k=1}^{\infty}$ such that $i \in \cap_{k=1}^\infty I^{n_k}$ and $\lim\limits_{k \to \infty} \bP (g^{n_k}_i)^{-1} = \bP g_i^{-1}$,
\item $\lim\limits_{n \to \infty} \int_I \bP (g^n_i)^{-1} \dif\lambda(i) = \int_I \bP g_i^{-1} \dif\lambda(i)$.
\end{enumerate}

We first prove the properties of aggregate robustness and admissibility. Then we prove the ex post robust perfection. Lastly, we show the refinement result.

\paragraph{Aggregate robustness.} Fix a player $i \in I^0$. We will construct a sequence of perturbation functions $\phi^t \colon \Delta \to \ehat{\cM}(\Delta)$ such that the two conditions in Definition~\ref{defn:agg-robust} are satisfied. Recall $\{g^{n_k}\}_{k=1}^{\infty}$ is the corresponding subsequence in Condition (ii). For each $k \in \bZ_+$, the ELLN (Lemma~\ref{lem:ELLN}) implies that for $\bP$-almost all $\omega \in \Omega$,
$$ s(g^{n_k}_\omega)  =  \lambda (g^{n_k}_\omega)^{-1}  \xlongequal{\text{ELLN}}  \int_I \bP (g^{n_k}_i)^{-1} \dif\lambda(i). $$
For each $k \in \bZ_+$, denote $\int_I \bP (g^{n_k}_i)^{-1} \dif\lambda(i)$ by $s'(g^{n_k})$. Similarly, we have $s(g_\omega) = \lambda g_\omega^{-1} = \int_I \bP g_i^{-1} \dif\lambda(i)$ for $\bP$-almost all $\omega \in \Omega$. Denote $\int_I \bP g_i^{-1} \dif\lambda(i)$ by $s'(g)$.

Fix an action $a \in \supp \bP g_i^{-1}$. Then we have $\bP g_i^{-1}(\{a\}) > 0$. Since $\lim\limits_{n \to \infty} \bP (g^{n_k}_i)^{-1}(\{a\}) = \bP g_i^{-1}(\{a\})$ and $\lim\limits_{k \to \infty} \varepsilon_{n_k} = 0$, there exists a positive integer $M_a$ such that for any $k > M_a$, $\bP (g^{n_k}_i)^{-1}(\{a\}) > \varepsilon_{n_k}$. Note that $g^{n_k}$ is an $\varepsilon_{n_k}$-robust perfect equilibrium with the perturbed player space $I^{n_k}$ and $i \in I^{n_k}$. Therefore, for each $k > M_a$ and for each $a' \in A$, we have
$$ \int_\Omega  u_i \Bigl( a, \ehat{s(g^{n_k}_\omega)} \Bigr) \dif\bP(\omega)  \ge  \int_\Omega  u_i \Bigl( a', \ehat{s(g^{n_k}_\omega)} \Bigr) \dif\bP(\omega), $$
where we use $\ehat{s(g^{n_k}_\omega)}$ to denote $\varphi^{n_k} \bigl(s(g^{n_k}_\omega)\bigr)$ for notational simplicity. Since $s(g^{n_k}_\omega) = s'(g^{n_k})$ for $\bP$-almost all $\omega$, we obtain that for each $k > M_a$ and for each $a' \in A$,
$$ u_i \Bigl( a, \ehat{s'(g^{n_k})} \Bigr)  \ge  u_i \Bigl( a', \ehat{s'(g^{n_k})} \Bigr). $$
Let $M = \max \{ M_a \mid a \in \supp \bP g_i^{-1} \}$. It is clear that for each $k > M$, for each $a \in \supp \bP g_i^{-1}$, and for each $a' \in A$,
$$ u_i \Bigl( a, \ehat{s'(g^{n_k})} \Bigr)  \ge  u_i \Bigl( a', \ehat{s'(g^{n_k})} \Bigr). $$

Therefore, for each $k > M$, for each $\xi \colon \Omega \to A$, and for $\bP$-almost all $\omega \in \Omega$,
$$ u_i \Bigl( g_i(\omega), \ehat{s'(g^{n_k})} \Bigr)  \ge  \int_\Omega u_i \Bigl( \xi(\omega), \ehat{s'(g^{n_k})} \Bigr) \dif\bP(\omega). $$
Taking integration on $\omega$, we have that for each $k > M$ and for each $\xi \colon \Omega \to A$,
$$ \int_\Omega u_i \Bigl( g_i(\omega), \ehat{s'(g^{n_k})} \Bigr) \dif\bP(\omega)  \ge  \int_\Omega u_i \Bigl( \xi(\omega), \ehat{s'(g^{n_k})} \Bigr) \dif\bP(\omega). $$

For any integer $k > M$, let $\phi^k(s'(g)) = \varphi^{n_k}(s'(g^{n_k})) = \ehat{s'(g^{n_k})}$ and $\phi^k(\tau) = \varphi^{n_k}(\tau)$ for each $\tau \ne s'(g)$. Then for each $k > M$ and for each $\xi \colon \Omega \to A$,
$$ \int_\Omega u_i \Bigl( g_i(\omega), \phi^k \bigl(s'(g)\bigr) \Bigr) \dif\bP(\omega)  \ge  \int_\Omega u_i \Bigl( \xi(\omega), \phi^k \bigl(s'(g)\bigr) \Bigr) \dif\bP(\omega). $$
Since $s'(g) = s(g_\omega)$ for $\bP$-almost all $\omega \in \Omega$, we have that for each $k > M$ and each $\xi \colon \Omega \to A$,
$$ \int_\Omega u_i \Bigl( g_i(\omega), \phi^k \bigl(s(g_\omega)\bigr) \Bigr) \dif \bP(\omega) \ge \int_\Omega u_i \Bigl( \xi(\omega), \phi^k \bigl(s(g_\omega)\bigr) \Bigr) \dif \bP(\omega). $$
According to the construction of $\phi^k$, we also have $\lim\limits_{k \to \infty} \phi^k(\tau) = \delta_\tau$ for each $\tau \in \Delta$. Therefore, $g$ is aggregate robust.

\paragraph{Admissibility.} Instead of proving this property directly, we show that admissibility is a consequence of aggregate robustness.

Fix a player $i \in I^0$. Assume that for this player, $a$ is weakly dominated by a mixed strategy $\xi \colon \Omega \to A$,
$$ u_i(a, \tau) \le  u_i(\xi, \tau) \text{ for all $\tau \in \Delta$, and } u_i(a, \tau')  <  u_i(\xi, \tau') \text{ for some $\tau' \in \Delta$,} $$
where $u_i(\xi ,\tau) = \int_\Omega u_i \bigl( \xi(\omega), \tau \bigr) \dif\bP(\omega)$.

For each $n \in \bZ_+$ and for each $\omega \in \Omega$, since $\ehat{s(g_\omega)} = \varphi^{n}(s(g_\omega))$ is a full-support probability distribution on $\Delta$ and $u_i$ is continuous on $A \times \Delta$, we have that
$$ u_i \Bigl( a, \ehat{s(g_\omega)} \Bigr)  =  \int_\Delta u_i(a, \tau) \dif \ehat{s(g_\omega)}(\tau)  <  \int_\Delta u_i(\xi, \tau) \dif \ehat{s(g_\omega)}(\tau)  =  u_i \Bigl( \xi, \ehat{s(g_\omega)} \Bigr). $$
By taking integration for $\omega$, we have that
$$ \int_\Omega u_i \Bigl(a, \ehat{s(g_\omega)} \Bigr) \dif\bP(\omega)  <  \int_\Omega u_i \Bigl( \xi, \ehat{s(g_\omega)} \Bigr) \dif\bP(\omega). $$
The above equation means that $a$ is not a best response to the perturbed societal summary $\ehat{s(g_\omega)}$. Since $g$ is aggregate robust, we conclude that $\bP g_i^{-1} (a) = 0$.

\paragraph{Ex post robust perfection.} Next we show that $g$ has the ex post robust perfection property, i.e., $g_\omega$ is a pure strategy robust perfect equilibrium for $\bP$-almost all $\omega \in \Omega$.

For each $n \in \bZ_+$ and $i \in I$, let $h^n(i) = \bP (g_i^n)^{-1}$. According to the proof of Proposition~\ref{prop:equivalence}, $h^n$ is an $\varepsilon_n$-robust perfect equilibrium associated with $I^n$. For simplicity, we assume that $\varepsilon_n < 1$ for any $n \in \bZ_+$. Based on Lemma~\ref{lem:universality}, we can assume that for $\lambda$-almost all $i \in I$, the sequence of strategies $\{g_i^n\}_{n \in \bZ_+} \colon \Omega \to A$ are independent.\footnote{One can construct an independent sequence $\{g^n_i\}_{n \in \bZ_+}$ such that induced distribution of each $g^n_i$ is $h^n(i)$ as follows. Let $H$ be a mapping from $I$ to $\Delta(A^{\bZ_+})$ defined by $H(i) = h^1(i) \otimes h^2(i) \otimes \cdots \otimes h^n(i) \cdots$. Then by Lemma~\ref{lem:universality}, there exists a mapping $(g^1, g^2, \cdots, g^n, \cdots)$ from $I \times \Omega$ to $A^{\bZ_+}$ such that its induced distribution is $H(i)$. Thus, we obtain the required sequence $\{g^n\}_{n \in \bZ_+}$.} We divide the proof into two steps. In Step 1, we construct a sequence of randomized strategy profiles $\{f_\omega^n\}_{n=1}^\infty$ such that each $f_\omega^n$ is a $2K\varepsilon_n$-robust perfect equilibrium for $\bP$-almost all $\omega \in \Omega$. In Step 2, for $\bP$-almost all $\omega \in \Omega$, the pure strategy profile $g_\omega$ is shown to be the ``limit'' of the sequence $\{f^n_\omega\}_{n \in \bZ_+}$.

\subparagraph{Step 1.}
Fix any $n \in \bZ_+$. For each $\omega \in \Omega$, we define a strategy profile $f_\omega^n$ such that
$$ f_\omega^n(i) = (1-\varepsilon_n) \delta_{g^n_\omega(i)} + \varepsilon_n h^n(i) \text{ for each $i \in I$}. $$

Since each $h^n$ is an $\varepsilon_n$-robust perfect equilibrium associated with $I^n$, for each $i \in I^n$, for all $a_m$, $a_\ell \in A$,
$$ u_i \bigl( a_m, \ehat{s(h^n)} \bigr) < u_i \bigl( a_\ell, \ehat{s(h^n)} \bigr)  \Rightarrow  h^n(i;a_m) \le \varepsilon_n. $$
It is clear that
$$ h^n \Bigl( i; \Br_i \bigl( \ehat{s(h^n)} \bigr) \Bigr) \ge 1 - K\varepsilon_n, $$
where $\Br_i(\ehat{s(h^n)})$ is the set of player $i$'s best responses given the perturbed societal summary $\ehat{s(h^n)}$.

By the ELLN (see Lemma~\ref{lem:ELLN}), there exists $\Omega^1 \subseteq \Omega$ with $\bP(\Omega^1)=1$ such that for any $\omega \in \Omega^1$,
$$ \lambda (g_\omega^n)^{-1}  \xlongequal{\text{ELLN}}  \int_I \bP (g^n_i)^{-1} \dif\lambda(i)  =  \int_I h^n(i) \dif\lambda(i). $$
Therefore, for any $\omega \in \Omega^1$,
$$ \lambda \Bigl( \Bigl\{ i \in I^n \;\Big\vert\; g_\omega^n(i) \in \Br_i \bigl( \ehat{s(h^n)} \bigr) \Bigr\} \Bigr)  =  \int_{I^n} h^n \Bigl( i; \Br_i \bigl( \ehat{s(h^n)} \bigr) \Bigr) \dif\lambda(i)  \ge  (1-K\varepsilon_n)^2, $$
where the inequality is due to that $h^n \bigl( i; \Br_i \bigl( \ehat{s(h^n)} \bigr) \bigr)  \ge  1 - K\varepsilon_n$ and $\lambda(I^n) \ge 1 - \varepsilon_n \ge 1 - K\varepsilon_n$. For any $\omega \in \Omega$, let $I^n_\omega = \{i \in I^n \mid g_\omega^n(i) \in \Br_i(\ehat{s(h^n)})\}$. Then $\lambda(I^n_\omega) \ge (1 - K\varepsilon_n)^2 \ge 1 - 2K\varepsilon_n$ for any $\omega \in \Omega^1$.

It is easy to see that each $f_\omega^n$ has a full support on $A$ by its construction. Moreover, the ELLN (see Lemma~\ref{lem:ELLN}) implies that for any $\omega \in \Omega^1$,
$$ s(f_\omega^n)  =  \int_I f_\omega^n(i) \dif\lambda(i)  =  (1-\varepsilon_n) \int_I \delta_{g^n_\omega(i)} \dif\lambda(i) + \varepsilon_n \int_I h^n(i) \dif\lambda(i)  \xlongequal{\text{ELLN}}  \int_I h^n(i) \dif\lambda(i)  =  s(h^n). $$
Thus, for each $\omega \in \Omega^1$ and $i \in I^n_\omega$, we know that
\begin{align*}
f_\omega^n \Bigl( i; \Br_i \bigl( \ehat{s(f_\omega^n)} \bigr) \Bigr)
    & = f_\omega^n \Bigl(i; \Br_i \bigl( \ehat{s(h^n)} \bigr) \Bigr) = (1-\varepsilon_n) \delta_{g^n_\omega(i)} \Bigl( \Br_i \bigl( \ehat{s(h^n)} \bigr) \Bigr) + \varepsilon_n h^n\Bigl(i; \Br_i \bigl( \ehat{s(h^n)} \bigr) \Bigr)   \\
    & \ge (1-\varepsilon_n) + \varepsilon_n (1-K\varepsilon_n)  =  1 - K\varepsilon^2_n \ge 1 - 2K\varepsilon_n.
\end{align*}
It is clear that for any $\omega \in \Omega^1$, $\lambda(I^n_\omega) \ge 1 - 2K\varepsilon_n$ and for any $i \in I^n_\omega$, $a_m$ and $a_\ell$ in $A$, if $u_i \bigl( a_m, \ehat{s(f_\omega^n)} \bigr) < u_i \bigl( a_\ell, \ehat{s(f_\omega^n)} \bigr)$, then $f_\omega^n(i; a_m) \le 2K\varepsilon_n$. Therefore, for any $\omega \in \Omega^1$, $f_\omega^n$ is a $2K\varepsilon_n$-robust perfect equilibrium associated with $I^n_\omega$. By taking countable intersections, we can find $\Omega^2 \subseteq \Omega$ with $\bP(\Omega^2)=1$ such that for any $\omega \in \Omega^2$, $f_\omega^n$ is a $2K\varepsilon_n$-robust perfect equilibrium associated with $I^n_\omega$ for any $n \in \bZ_+$.

\subparagraph{Step 2.}
Let $h$ be a randomized strategy profile such that $h(i) = \bP g_i^{-1}$ for each $i \in I$. Recall that $I^0$ is the set of players such that $\lambda(I^0) = 1$ and for each $i \in I^0$, there exists a subsequence $\{h^{n_k}\}_{k=1}^\infty$ such that $i \in \cap_{k=1}^\infty I^{n_k}$ and
$$ \lim\limits_{k \to \infty} h^{n_k}(i)  =  \lim\limits_{k \to \infty} \bP (g^{n_k}_i)^{-1}  =  \bP g_i^{-1}  =  h(i). $$

Fix any $i \in I^0$ and $a \in \supp h(i)$. Since $\lim\limits_{k \to \infty} h^{n_k}(i) = h(i)$, $h^{n_k} \bigl( i; \Br_i(\ehat{s( h^{n_k})}) \bigr)  \ge  1 - K\varepsilon_{n_k}$ and $\lim\limits_{k \to \infty} \varepsilon_{n_k} = 0$, there exists $M \in \bZ_+$ such that for $k \ge M$,
\begin{align*}
\bP (g^{n_k}_i)^{-1} \Bigl( \{a\} \cap \Br_i \bigl( \ehat{s(h^{n_k})} \bigr) \Bigr) & =  h^{n_k} \Bigl( i; \{a\} \cap \Br_i \bigl( \ehat{s(h^{n_k})} \bigr) \Bigr)   \\
& \ge h^{n_k}(i; \{a\}) - h^{n_k} \Bigl( i; A \setminus \Br_i \bigl( \ehat{s(h^{n_k})} \bigr) \Bigr) >  \tfrac{1}{2} h(i; \{a\})  >  0.
\end{align*}
For each $k \in \bZ_+$, let $\Omega^{n_k} = \bigl\{ \omega \in \Omega \mid g_i^{n_k}(\omega) \in \{a\} \cap \Br_i \bigl( \ehat{s(h^{n_k})} \bigr) \bigr\}$. Since $\{g_i^n\}_{n \in \bZ_+}$ are independent, we have that the events $\{ \Omega^{n_k} \}_{k \ge M}$ are independent. Furthermore, $\sum_{k \ge M} \bP(\Omega^{n_k}) = \infty$. Then the second Borel-Cantelli lemma (see Theorem 2.3.6 in \cite{Durrett2010}) implies that for $\bP$-almost all $\omega \in \Omega$,
$$ g_i^{n_k}(\omega) \in \{a\} \cap \Br_i \bigl( \ehat{s(h^{n_k})} \bigr) \text{ infinitely many times.} $$
Therefore, for $\bP$-almost all $\omega \in \Omega$, there exists a subsequence $\{n_{k_t}\}_{t=1}^\infty$ such that for each $t \in \bZ_+$, $g_i^{n_{k_t}}(\omega) = a$ and $g_i^{n_{k_t}}(\omega) \in \Br_i \bigl( \ehat{s(h^{n_{k_t}})} \bigr)$. Since $h(i) = \bP g_i^{-1}$, it is clear that for $\bP$-almost all $\omega \in \Omega$, $g_i(\omega) \in \supp h(i)$. Therefore, there exists $\Omega_i \subseteq \Omega$ with $\bP(\Omega_i)=1$ such that for any $\omega \in \Omega_i$, $g_i(\omega) \in \supp h(i)$ and
 there exists a subsequence $\{n_l\}_{l=1}^\infty$ such that for each $l \in \bZ_+$, $g_i^{n_l}(\omega) = g_i(\omega)$, $i \in I^{n_l}$ and $g_i^{n_l}(\omega) \in \Br_i \bigl( \ehat{s(h^{n_l})} \bigr)$, which also lead to $i \in I_\omega^{n_l}$.

Let $e$ be a mapping from $I \times \Omega$ to $A$ such that
$$ e(i, \omega) = \begin{cases}
a_1,    & \text{if } g(i, \omega) \ne a_1,  \\
a_2,    & \text{if } g(i, \omega) = a_1.
\end{cases} $$
It is clear that $e$ is $\cI \boxtimes \cF$-measurable and $e(i, \omega) \ne g(i,\omega)$ for any $(i,\omega) \in I \times \Omega$. For any $n \in \bZ_+$, let $e^n$ be an $\cI \boxtimes \cF$-measurable mapping from $I \times \Omega$ to $A$ such that
$$ e^n(i, \omega) = \begin{cases}
g^n(i, \omega), & \text{if } i \in I^n_\omega,  \\
e(i, \omega),   & \text{if } i \notin I^n_\omega.
\end{cases} $$
We use $E$ to denote the $\cI \boxtimes \cF$-measurable set $\bigl\{ (i, \omega) \in I \times \Omega \mid g(i, \omega) \in \cap_{k=1}^{\infty} \cup_{r=k}^{\infty} \{e^r(i, \omega)\} \bigr\}$. By the result in end of the last paragraph, it is clear that for any $i \in I^0$, $\Omega_i \subseteq E_i$, which implies that $\bP(E_i)=1$. Note that $\lambda(I^0)=1$. By the Fubini property, we know that there exists $\Omega^3 \subseteq \Omega$ with $\bP(\Omega^3)=1$ such that for any $\omega \in \Omega^3$, $\lambda(E_\omega)=1$. By the construction of $e$ and $e^n$, for any $\omega \in \Omega$ and any $i \in E_\omega$, it is impossible to find $T \in \bZ_+$ such that $i \notin I^n_\omega$ for each $n \ge T$; otherwise, $e^n(i, \omega) = e(i, \omega)$ for each $n \ge T$, and hence $g(i, \omega) \notin \cap_{k=1}^\infty \cup_{r=k}^\infty \{e^r(i, \omega)\}$, which is a contradiction. Thus, for any $\omega \in \Omega^3$ and any $i \in E_\omega$, there exists a subsequence $\{n_r\}_{r=1}^\infty$ such that for each $r \in \bZ_+$, $i \in I^{n_r}_\omega$ and $\lim\limits_{r \to \infty} g_i^{n_r}(\omega) = g_i(\omega)$.

Since $f_\omega^n(i) = (1-\varepsilon_n) \delta_{g^n_\omega(i)} + \varepsilon_n h^n(i)$, it is clear that
$$ \lim_{r \to \infty} f_\omega^{n_r}(i)  =  \lim_{r \to \infty} \delta_{g^{n_r}_\omega(i)}  =  \delta_{g_\omega(i)}. $$
By the ELLN (see Lemma~\ref{lem:ELLN}), there exists $\Omega^4 \subseteq \Omega$ with $\bP(\Omega^4)=1$ such that for any $\omega \in \Omega^4$,
$$ \lim\limits_{n \to \infty} \int_I f_\omega^n(i) \dif\lambda  =  \lim_{n \to \infty} \int_I h^n(i) \dif\lambda  =  \lim_{n \to \infty} \int_I \bP (g^n_i)^{-1} \dif\lambda  =  \int_I \bP g_i^{-1} \dif\lambda  \xlongequal{\text{ELLN}}  \int_I \delta_{g_\omega(i)} \dif\lambda. $$
Note that for any $\omega \in \Omega^2$, $f_\omega^n$ is a $2K\varepsilon_n$-robust perfect equilibrium associated with $I^n_\omega$ for any $n \in \bZ_+$. Hence, for any $\omega \in \Omega^2 \cap \Omega^3 \cap \Omega^4$, Conditions (i), (ii) and (iii) in Definition~\ref{defn:mixPE} hold for $g_\omega$. It is clear that $\bP(\Omega^2 \cap \Omega^3 \cap \Omega^4) = 1$. Therefore, for $\bP$-almost all $\omega \in \Omega$, the ex post strategy profile $g_\omega$ is a pure strategy robust perfect equilibrium.

\paragraph{Refinement.} Finally we show that $g$ is a mixed strategy Nash equilibrium. Fix a player $i \in I^0$ and an action $a \in A$ with $\bP g^{-1}_i(a) > 0$. Recall $\{g^{n_k}\}_{k=1}^{\infty}$ is the corresponding subsequence in Condition (ii). Then we know that $\lim\limits_{k \to \infty} \bP (g^{n_k}_i)^{-1}(a)  =  \bP g_i^{-1}(a) > 0$.
Hence, there exists $L \in \bZ_+$ such that for any $k \geq L$, $\bP (g^{n_k}_i)^{-1}(a) > \varepsilon_{n_k}$. By Condition (ii) in Definition~\ref{defn:e-mixPE},
$$ \int_\Omega u_i \Bigl( a, \ehat{s(g_\omega^{n_k})} \Bigr) \dif \bP(\omega)  \ge  \int_\Omega u_i \Bigl( a', \ehat{s(g_\omega^{n_k})} \Bigr) \dif \bP(\omega) \text{ for any } a'\in A. $$
Note that for $\bP$-almost all $\omega \in \Omega$, $s(g^{n_k}_\omega) = s'(g^{n_k})$. It is clear that $u_i \Bigl( a, \ehat{s'(g^{n_k})} \Bigr) \ge u_i \Bigl( a', \ehat{s'(g^{n_k})} \Bigr)$ for any $a' \in A$ and any $k \ge L$. Since $u_i$ is continuous on $A \times \Delta$ and $\lim\limits_{k \to \infty} \ehat{s'(g^{n_k})} = \delta_{s'(g)}$,
\begin{align*}
\int_\Omega u_i \bigl( a, s(g_\omega) \bigr) \dif \bP(\omega)
    & = u_i \Bigl( a, \delta_{s'(g)} \Bigr) = \lim\limits_{k \to \infty} u_i \Bigl( a, \ehat{s'(g^{n_k})} \Bigr)   \\
    & \ge \lim\limits_{k \to \infty} u_i \Bigl( a', \ehat{s'(g^{n_k})} \Bigr) = u_i \Bigl( a', \delta_{s'(g)}) \Bigr) = \int_\Omega u_i \bigl( a', s(g_\omega) \bigr) \dif \bP(\omega),
\end{align*}
for any $a' \in A$. Since $a$ is arbitrarily fixed in the support of $\bP g^{-1}_i$, we know that $g$ is a mixed strategy Nash equilibrium.

\subsection{Proof of Theorem~\ref{thm:exist}}
\label{sec:appen:exist}

\paragraph{(1)} To prove the existence of a symmetric robust perfect equilibrium, we consider a new large game. The new player space $(\bar{I}, \bar{\cI}, \bar{\lambda})$ defined as follows: we identify $i$ and $i'$ in $I$ by the equivalence relation $i \sim i'$ if and only if $u_i = u_{i'}$, let $\bar{I}$ denote the quotient space $I/\sim$, and let the quotient map be $\pi \colon I \to \bar{I}$. The measure structure on $\bar{I}$ is defined as the pushforward under $\pi$ of the measure structure on $I$, i.e., a subset $J$ of $\bar{I}$ is in $\bar{\cI}$ if $\pi^{-1}(J)$ is in $\cI$, and $\bar{\lambda}(J) = \lambda\bigl(\pi^{-1}(J)\bigr)$ for any $J \in \bar{\cI}$. For any $\bar{i} \in \bar{I}$, $u_{\bar{i}}$ is defined to be the payoff function of any player in $\pi^{-1}(\bar{i})$. According to the definition of $\pi$, it is clear that $u_{\bar{i}}$ is well defined.

Fix any $\varepsilon \in (0, 1)$. Let $N^\varepsilon$ be the set of measurable functions from $\bar{I}$ to $\Delta$ such that for each $h \in N^\varepsilon$, we have $h(\bar{i}; a)\ge \frac{\varepsilon}{K}$ for $\bar{\lambda}$-almost all $\bar{i} \in \bar{I}$ and every $a \in A$. Let $\Delta^\varepsilon = \{\int_{\bar{I}} h(\bar{i})\dif\bar{\lambda}(\bar{i})\mid h\in N^\varepsilon\} =   \{(x_1, x_2, \ldots, x_K) \mid x_1 + x_2 + \cdots + x_K = 1, x_k \ge \frac{\varepsilon}{K}, \forall k = 1, 2, \dots, K\}.$ Clearly, $\Delta^\varepsilon$ is a nonempty convex and compact subset of $\Delta$. We define a perturbation function $\varphi(\tau) = (1-\varepsilon)\delta_\tau+\varepsilon\eta$ for any $\tau \in \Delta$, where $\eta$ is the uniform distribution on $\Delta$. For simplicity, $\varphi(\tau)$ is written as $\ehat{\tau}$. For each $\bar{i} \in \bar{I}$ and $\tau \in \Delta^\varepsilon$, let
$$ D(\bar{i}, \tau) = (1-\varepsilon)\argmax\limits_{\mu \in \Delta} u_{\bar{i}}(\mu,\ehat{\tau})+\varepsilon \nu := \Bigl\{ (1-\varepsilon) \mu' + \varepsilon \nu \;\Big\vert\; \mu' \in \argmax_{\mu \in \Delta} u_{\bar{i}}(\mu, \ehat{\tau}) \Bigr\}, $$
where $\nu \in \Delta$ and each component of $\nu$ is $\frac{1}{K}$, and $u_{\bar{i}}(\mu,\ehat{\tau}) = \int_\Delta\int_A u_{\bar{i}}(a, \tau')\dif\mu(a)\dif\ehat{\tau}(\tau')$. Clearly, the correspondent $D$ is convex-valued. Moreover, Berge's Maximum Theorem implies that $D(\bar{i}, \cdot) \colon \Delta^\varepsilon \twoheadrightarrow \Delta^\varepsilon$ is nonempty, compact-valued, and upper hemicontinuous for each $\bar{i} \in \bar{I}$; $D(\cdot, \tau)\colon \bar{I} \twoheadrightarrow \Delta^\varepsilon$ is measurable for each $\tau$ due to the Measurable Maximum Theorem. Define a correspondence $B \colon \Delta^\varepsilon \twoheadrightarrow \Delta^\varepsilon$ as follows:
$$ B(\tau) = \int_{\bar{I}} D(\bar{i}, \tau) \dif\bar{\lambda}(\bar{i}). $$
Since $D$ is convex-valued, we have that $B(\tau)$ is convex-valued. Moreover, $B \colon \Delta^\varepsilon \twoheadrightarrow \Delta^\varepsilon$ is upper hemicontinuous; see \citet[p.~68]{Hildenbrand1974}. By Kakutani fixed point theorem, $B$ has a fixed point $\tau^\varepsilon$. That is, there is an $h^\varepsilon \in N^\varepsilon$ such that $\int_{\bar{I}} h^\varepsilon \dif\bar\lambda = \tau^\varepsilon$ and $h^\varepsilon(\bar{i}) \in D(\bar{i}, \tau^\varepsilon)$ for each $\bar{i} \in \bar{I}$. According to the construction of $(\bar{I}, \bar{\cI}, \bar{\lambda})$, we have that
$$ \int_I h^\varepsilon\circ\pi\dif\lambda = \int_{\bar{I}} h^\varepsilon\dif\bar{\lambda} =\tau^\varepsilon. $$
For any $i \in I$, we know that
$$ h^\varepsilon \bigl( \pi(i) \bigr) \in D \bigl( \pi(i), \tau^\varepsilon \bigr) = (1-\varepsilon)\argmax\limits_{\mu \in \Delta} u_i(\mu,\ehat{\tau^\varepsilon})+\varepsilon \nu. $$
Hence, $h^\varepsilon\circ\pi$ is an $\varepsilon$-robust perfect equilibrium associated with the perturbed player space $I$ and perturbation function $\varphi$.

Let $\{\varepsilon_n\}_{n\in \bZ_+}$ be a sequence such that $\varepsilon_n \in (0, 1)$ and $\varepsilon_n \to 0$. For each $\varepsilon_n$, there exists a measurable function $h^n$ from $\bar{I}$ to $\Delta$ such that $h^n \circ \pi$ is an $\varepsilon_n$-robust perfect equilibrium associated with the perturbed player space $I$. By Lemma 3 in \citet[p.~69]{Hildenbrand1974}, there is a randomized strategy profile $h$ such that $\lim\limits_{n \to \infty} \int_{\bar{I}} h^n \dif\bar\lambda = \int_{\bar{I}} h \dif\bar\lambda$ and for $\bar{\lambda}$-almost all $\bar{i} \in \bar{I}$, $h(\bar{i}) = \lim\limits_{k \to \infty}h^{n_k}(\bar{i})$ for a subsequence $\{h^{n_k}\}_{k=1}^\infty$. Since $\lim\limits_{n \to \infty}\int_{I} h^n\circ\pi\dif\lambda = \lim\limits_{n \to \infty}\int_{\bar{I}} h^n\dif\bar{\lambda} = \int_{\bar{I}} h\dif\bar{\lambda} = \int_I h\circ\pi\dif\lambda$, and for $\lambda$-almost all $i \in I$, $h\circ\pi(i) = h\bigl(\pi(i)\bigr) = \lim\limits_{k \to \infty}h^{n_k}\bigl(\pi(i)\bigr) = \lim\limits_{k \to \infty}h^{n_k}\circ\pi(i)$ for a subsequence $\{h^{n_k}\}_{k=1}^\infty$, we conclude that $h\circ\pi$ is a symmetric randomized robust perfect equilibrium. Finally, Proposition~\ref{prop:equivalence} implies that the symmetric randomized strategy robust perfect equilibrium $h\circ\pi$ can be lifted to a symmetric mixed strategy robust perfect equilibrium.

\paragraph{(2)} From part (1), we have a symmetric mixed strategy robust perfect equilibrium $g$. Then Theorem~\ref{thm:robust} implies that (almost) every realization $g_\omega$ of $g$ is a pure strategy robust perfect equilibrium. This concludes the proof.

\subsection{Proof of Claim~\ref{claim:example:non-robust}}
\label{sec:appen:example}

First we show that none of $\{a, b, c\}$ is weakly dominated. For an arbitrary action $p \in \{a, b, c\}$, and any mixed strategy $\xi$ that is not the pure strategy $p$, let $\tau$ be a distribution such that $\tau(p) = 1$. Then we have $u(p, \tau) > u(\xi, \tau)$, which implies that $p$ is not weakly dominated by $\xi$. Moreover, given the strategy profile $f$ and its induced action distribution $\tau^* = (\frac{1}{3}, \frac{1}{3}, \frac{1}{3})$, $a$, $b$ and $c$ are indifferent for each player. Thus, $f$ is an admissible Nash equilibrium.

Next we prove that $f$ satisfies the aggregate robustness. Fix a player $i$, suppose $f(i) = p$. We need to construct a sequence of perturbation functions $\{\varphi^n\}_{n \in \bZ_+}$ such that $p$ is always a best response to those perturbed societal summaries $\{\varphi^n(\tau^*)\}_{n \in \bZ_+}$. Pick a sequence $\{\varepsilon_n\}_{n \in \bZ_+}$ in $(0, \tfrac{1}{3})$ such that $\lim\limits_{n \to \infty} \varepsilon_n = 0$. We define each $\varphi^n$ as follows:
$$\varphi^n(\tau) = (1 - \varepsilon_n) \delta_\tau + (\varepsilon_n - \varepsilon_n^2) \delta_p + \varepsilon_n^2\eta,$$
where $\eta$ is the uniform distribution on $\Delta$. Clearly, $\varphi^n(\tau)$ weakly converges to $\delta_\tau$. Then we show that $p$ is a best response. For any action $p'\ne p$, we have
\begin{align*}
u\bigl(p, \varphi^n(\tau^*)\bigr) - u\bigl(p', \varphi^n(\tau^*)\bigr) ={} &(1 - \varepsilon_n) u(p, \tau^*) + (\varepsilon_n - \varepsilon_n^2) u(p, \delta_p) + \varepsilon_n^2 \int_\Delta u(p, \tau') \dif \eta(\tau') \\
&- (1 - \varepsilon_n) u(p', \tau^*) - (\varepsilon_n - \varepsilon_n^2) u(p', \delta_p) - \varepsilon_n^2 \int_\Delta u(p', \tau') \dif \eta(\tau').
\end{align*}
It is easy to see that $u(p, \tau^*) = u(p', \tau^*)$. In addition, according to the construction of $u$ we know that $u(p, \delta_p) - u(p', \delta_p) \ge 1$. Moreover, since $-1 \le u(p, \tau') \le 1$ for any $\tau'$, we conclude that $\int_\Delta u(p, \tau') \dif \eta(\tau') - \int_\Delta u(p', \tau') \dif \eta(\tau') \ge -2$. Therefore, we have
\begin{align*}
u\bigl(p, \varphi^n(\tau^*)\bigr) - u\bigl(p', \varphi^n(\tau^*)\bigr) \ge (\varepsilon_n - \varepsilon_n^2) - 2 \varepsilon_n^2 = \varepsilon_n - 3 \varepsilon_n^2 > 0.
\end{align*}
Hence $p$ is a best response to each perturbed societal summary $\varphi^n(\tau^*)$, and $f$ satisfies the aggregate robustness.

Finally we show that $f$ is not a robust perfect equilibrium. Otherwise suppose that $f$ is a robust perfect equilibrium. Then according to Proposition~\ref{prop:equivalence}, there exists a sequence of randomized strategy profiles $\{h^n\}_{n \in \bZ_+}$ and a sequence of positive constants $\{\varepsilon_n\}_{n \in \bZ_+}$ such that
\begin{enumerate}[label=(\roman*)]
\item each $h^n$ is an $\varepsilon_n$-robust perfect equilibrium (each $h^n$ is associated with a perturbed player space $I_n$) with $\varepsilon_n \to 0$ as $n$ goes to infinity,
\item for $\lambda$-almost all $i \in I$, there exists a subsequence $\{h^{n_k}\}_{k=1}^\infty$ (each $h^{n_k}$ is associated with $I_{n_k}$) such that $i \in \cap_{k=1}^\infty I_{n_k}$ and $\lim\limits_{k \to \infty} h^{n_k}(i) = f(i)$,
\item $\lim\limits_{n \to \infty} \int_I h^n(i) \dif\lambda(i) = \int_I f(i) \dif\lambda(i) = \lambda f^{-1}$.
\end{enumerate}
For each $n \in \bZ_+$, let $\tau^n = \int_I h^n(i) \dif\lambda(i)$. Then $\lim\limits_{n \to \infty} \tau^n = \tau^* = (\frac{1}{3}, \frac{1}{3}, \frac{1}{3})$.

Fix an $\varepsilon_n$-robust perfect equilibrium $h^n$ and its induced action distribution $\tau^n$. We claim that there exists an action which is not a best response for all players with respect to $\ehat{\tau^n}$. If all three actions are best responses, then we have
$$ \int_\Delta u(a, \tau') \dif\ehat{\tau^n}(\tau')  =  \int_\Delta u(b, \tau') \dif\ehat{\tau^n}(\tau')  =  \int_\Delta u(c, \tau') \dif\ehat{\tau^n}(\tau'). $$

Since $u(a, \tau) = 0$ for each $\tau \in \Delta$, we have $0  =  \int_\Delta u(a, \tau') \dif\ehat{\tau^n}(\tau')  =  \int_\Delta u(b, \tau') \dif\ehat{\tau^n}(\tau')$. Notice that $u(b, \tau) = 0$ when $\tau(b) \in [0.3, 0.8]$. Thus, we have that
$$ \int_{\{\tau' \mid \tau'(b) < 0.3\}} u(b, \tau') \dif\ehat{\tau^n}(\tau')  +  \int_{\{\tau' \mid \tau'(b) > 0.8\}} u(b, \tau') \dif\ehat{\tau^n}(\tau')  =  0. $$
Therefore,
\begin{align*}
\int_{\{\tau' \mid \tau'(b) > 0.8\}} u(b, \tau') \dif\ehat{\tau^n}(\tau')
& = \int_{\{\tau' \mid \tau'(b) < 0.3\}} - u(b, \tau') \dif\ehat{\tau^n}(\tau')   \\
& = \int_{\{\tau' \mid \tau'(b) < 0.2\}} - u(b, \tau') \dif\ehat{\tau^n}(\tau') + \int_{\{\tau' \mid 0.2 \le \tau'(b) < 0.3\}} - u(b, \tau') \dif\ehat{\tau^n}(\tau')\\
& > \int_{\{\tau' \mid \tau'(b) < 0.2\}} - u(b, \tau') \dif\ehat{\tau^n}(\tau')
=  \int_{\{\tau' \mid \tau'(b) < 0.2\}} 1 \dif\ehat{\tau^n}(\tau')   \\
& \ge \int_{\{\tau' \mid \tau'(c) > 0.8\}} 1 \dif\ehat{\tau^n}(\tau')
\ge  \int_{\{\tau' \mid \tau'(c) > 0.8\}} u(c, \tau') \dif\ehat{\tau^n}(\tau'),
\end{align*}
where the first inequality is because $\ehat{\tau^n}$ has a full support on $\Delta$ and $- u(b, \tau')$ is positive on $\{\tau' \mid 0.2 \le \tau'(b) < 0.3\}$, and the second inequality is because the set $\{\tau' \mid \tau'(c) \ge 0.8\}$ is a subset of $\{\tau' \mid \tau'(b) \le 0.2\}$. So we conclude that
$$ \int_{\{\tau' \mid \tau'(b) > 0.8\}} u(b, \tau') \dif\ehat{\tau^n}(\tau')
>  \int_{\{\tau' \mid \tau'(c) > 0.8\}} u(c, \tau') \dif\ehat{\tau^n}(\tau'). $$

Similarly, we have that $0  =  \int_\Delta u(a, \tau') \dif\ehat{\tau^n}(\tau')  =  \int_\Delta u(c, \tau') \dif\ehat{\tau^n}(\tau')$, and hence
$$ \int_{\{\tau' \mid \tau'(c) > 0.8\}} u(c, \tau') \dif\ehat{\tau^n}(\tau')  >  \int_{\{\tau' \mid \tau'(b) > 0.8\}} u(b, \tau') \dif\ehat{\tau^n}(\tau'). $$
So we have a contradiction.

According to the definition of $\varepsilon_n$-robust perfect equilibrium, the probability of that non-best-response action is less than $2\varepsilon_n$. Thus, each limit point of $\{\tau^n\}_{n \in \bZ_+}$ should lie on the boundary of $\Delta$, and cannot be $\tau^*$. Therefore, $f$ satisfies the admissibility and aggregate robustness, but is not a robust perfect equilibrium.

\subsection{Proofs of the results in Section~\ref{sec:application}}
\label{sec:appen:application}

\begin{proof}[Proof of Proposition~\ref{prop:congestion-PE-exist}]
The existence of a symmetric mixed strategy robust perfect equilibrium and a pure strategy robust perfect equilibrium is shown in Theorem~\ref{thm:exist}. Let $G$ be an atomless congestion game with strictly increasing cost functions. Since every cost function is strictly increasing, there exists a unique traffic flow $\tau^0$ induced by pure strategy Nash equilibria; see \citet[Section~3.1.3, p.~3.13]{BMW1956} for more details. Let $g$ be a mixed strategy robust perfect equilibrium. By Theorem~\ref{thm:robust}, for $\bP$-almost all $\omega \in \Omega$, $g_{\omega}$ is a pure strategy robust perfect equilibrium, and hence a Nash equilibrium as well. Note that $\tau^0$ is the unique traffic flow induced by Nash equilibria. Hence, for $\bP$-almost all $\omega \in \Omega$, $s(g_{\omega})=\tau^0$. Therefore, by the exact law of large numbers, $\int_I \bP g_i^{-1} \dif\lambda(i) = \lambda g_\omega^{-1} = s(g_\omega) = \tau^0$ for $\bP$-almost all $\omega \in \Omega$, that is, the traffic flow induced by robust perfect equilibria is unique. Furthermore, if $g$ is symmetric, then we have $\bP g_i^{-1} = \int_I \bP g_i^{-1} \dif\lambda(i) = \tau^0$ for every driver $i$, which implies that $G$ has a unique symmetric robust perfect equilibrium.
\end{proof}

\begin{proof}[Proof of Proposition~\ref{prop:algorithm}]
Let $h^\varepsilon$ be the randomized strategy profile induced by $g^\varepsilon$. Then, $h^\varepsilon(i) = \sum\limits_{p \in \cP} \tau^\varepsilon(p) \delta_p$ for each player $i \in I$. The societal summary induced by $h^\varepsilon$ is $\int_I h^\varepsilon(i) \dif\lambda(i) = \tau^\varepsilon$. Let $\varphi$ be a perturbation function such that $\varphi(\tau) = (1-\varepsilon)\delta_\tau + \varepsilon \eta$ for each $\tau$, where $\eta$ is the uniform distribution on $\Delta$. In the following, we prove that $h^\varepsilon$ is an $\varepsilon$-robust perfect equilibrium (with the perturbed player space being $I$).

The optimization problem~\eqref{eq:optimization} can be rewritten as:
\begin{equation}
\label{eq:optimization'}
\tag{$**$}
\begin{aligned}
& \underset{\tau \in \Delta}{\text{minimize}}	& & \sum_{e \in E} \int_0^{\sum\limits_{e \in p}\tau(p)} C_e^{\varepsilon}(x) \dif x \\
& \text{subject to}								& & \sum\limits_{p \in \cP} \tau(p) = 1 \text{ and } \tau(p) \ge \varepsilon \text{ for all } p \in \cP.
\end{aligned}
\end{equation}
Since each cost function $C_e$ is increasing, each perturbed cost function $C_e^\varepsilon$ is also increasing, and hence the optimization problem is a convex problem. Thus, the Karush–Kuhn–Tucker conditions are necessary and sufficient. The Lagrangian function is
$$ L(\tau, \lambda, \mu) = \sum_{e \in E} \int_0^{\sum\limits_{e \in p}\tau(p)} C_e^\varepsilon(x) \dif x + \lambda \Bigl( \sum\limits_{p \in \cP} \tau(p) - 1 \Bigr) + \sum\limits_{p \in \cP} \mu_p \bigl( \varepsilon - \tau(p) \bigr), $$
where $\lambda$ and $\mu_p$ are the corresponding Lagrangian coefficients.

Since $\tau^\varepsilon = \bigl( \tau^\varepsilon(p) \bigr)_{p \in \cP}$ is a solution to the optimization problem~\eqref{eq:optimization}, it is also a solution to the optimization problem~\eqref{eq:optimization'}. Thus, it satisfies the Karush–Kuhn–Tucker conditions; that is, for each $p \in \cP$,
$$ \frac{\partial L}{\partial \tau(p)}(\tau^\varepsilon, \lambda, \mu)  =  \sum\limits_{e \in p} C_e^\varepsilon \bigl (\tau^\varepsilon(e) \bigr) + \lambda - \mu_p  =  0 $$
with the complementary slackness
$$ \mu_p \bigl( \varepsilon - \tau^\varepsilon(p) \bigr) = 0 \text{ and } \mu_p \ge 0. $$
Thus, for each $p \in \cP$, $\sum\limits_{e \in p} C_e^\varepsilon \bigl( \tau^\varepsilon(e) \bigr) \ge -\lambda$ and the equality holds when $\tau^\varepsilon(p) > \varepsilon$.

There exists a path $p \in \cP$ such that $\sum\limits_{e \in p} C_e^\varepsilon \bigl( \tau^\varepsilon(e) \bigr) = -\lambda$. Otherwise, $\tau^\varepsilon(p) = \varepsilon$ for each $p \in \cP$, and hence $1 = \sum_{p \in \cP} \tau^\varepsilon(p) = \sum_{p \in \cP} \varepsilon < \sum_{p \in \cP} |\cP| = 1$, which is a contradiction. Note that the path $p$ with $\sum\limits_{e \in p} C_e^\varepsilon \bigl( \tau^\varepsilon(e) \bigr) = -\lambda$ is a path with lowest perturbed cost. If the path $p$ is not a path with lowest perturbed cost, then $\sum\limits_{e \in p} C_e^\varepsilon \bigl( \tau^\varepsilon(e) \bigr) > -\lambda$, and hence $\mu_p > 0$ and $\tau^\varepsilon(p) = \varepsilon$. That is, each player $i$ chooses the path $p$ with probability $\tau^\varepsilon(p) = \varepsilon$ in the strategy $h^\varepsilon(i) = \sum\limits_{p \in \cP} \tau^\varepsilon(p) \delta_p$. So we conclude that $h^\varepsilon$ is an $\varepsilon$-robust perfect equilibrium in randomized strategies (with the perturbed player space $I$). By the same logic in the proof of Proposition~\ref{prop:equivalence}, $g^\varepsilon$ is an $\varepsilon$-robust perfect equilibrium with the perturbed player space being $I$.
\end{proof}

\begin{proof}[Proof of the statement in Remark~\ref{rmk:algorithm}]
Suppose that $f$ is a pure strategy profile with the action distribution $\tau^*$. We want to show that $f$ is a robust perfect equilibrium. Consider the randomized strategy profile $h$ with $h(i) = \sum\limits_{p \in \cP} \tau^* (p) \delta_p$ for each $i \in I$. According to the construction of $\tau^*$, we know that $\tau^*$ is a limit point of $ \{\tau^n\}_{n=1}^\infty$. Without loss of generality, we assume $\tau^n \to \tau^*$. For each $n \in \bZ_+$, let $g^n$ be the $\varepsilon_n$-robust perfect equilibrium obtained in Proposition~\ref{prop:algorithm}, and $h^n$ be the equivalent randomized strategy $\varepsilon_n$-robust perfect equilibrium. That is, $h^n(i) = \sum_{p \in \cP} \tau^n(p) \delta_p$ for each $i \in I$.

So we have $\lim\limits_{n \to \infty} h^n(i)= h(i)$ for each $i \in I$ and $\lim\limits_{n \to \infty} \int_I h^n(i) \dif\lambda(i) = \int_I h(i) \dif\lambda(i)$. Hence, $h$ is a (randomized strategy) robust perfect equilibrium. By Proposition~\ref{prop:equivalence}, $h$ can be lifted to a mixed strategy robust perfect equilibrium $g \colon I \times \Omega \to \cP$. Due to the ex post robust perfection property of $g$, we conclude that for $\bP$-almost all $\omega \in \Omega$, $g_\omega$ is a pure strategy robust perfect equilibrium. By the ELLN, the distribution induced by $g_\omega$ is $\int_I h(i) \dif\lambda(i) = \tau^*$. That is, $g_\omega$ and $f$ have the same distribution. Since all the players have the same action set and payoff function, $f$ is also a pure strategy robust perfect equilibrium.
\end{proof}

\begin{proof}[Proof of Proposition~\ref{prop:congestion-NEPE}]
Suppose that there are two paths $a$ and $b$, and $g$ is a mixed strategy admissible Nash equilibrium. We show that $g$ is a robust perfect equilibrium. Let $h$ be a randomized strategy profile induced by $g$, i.e., $h(i) = \bP g_i^{-1}$ for almost all $i$. According to Proposition~\ref{prop:equivalence}, we only need to show that $h$ is a robust perfect equilibrium. It is easy to check that $h$ is an admissible Nash equilibrium. Let $\tau^* = (\tau_a^*, \tau_b^*)$ be the societal summary of $h$, that is, $\tau^* = \int_I h(i) \dif\lambda(i)$.

If $u(a, \tau) = u(b, \tau)$ for each $\tau$, then $a$ and $b$ are indifferent for each player no matter what the distribution is. Thus, it is trivial that $h$ is a randomized strategy robust perfect equilibrium. In the following, we assume that the functions $u(a, \tau)$ and $u(b, \tau)$ are different. Our analysis will be divided into two cases.

\paragraph{Case 1.} If $\tau_a^*$ and $\tau_b^*$ are positive, then neither $a$ nor $b$ is a weakly dominated strategy, and we have $u(a, \tau^*) = u(b, \tau^*)$. Thus, there exist two distributions $\tau^1$ and $\tau^2$ in $\Delta$ such that $u(a, \tau^1) > u(b, \tau^1)$ and $u(a, \tau^2) < u(b, \tau^2)$. Therefore, there exists a full-support measure $\eta$ on $\Delta$ such that
$$ \int_\Delta u(a, \tau') \dif\eta(\tau')  =  \int_\Delta u(b, \tau') \dif\eta(\tau'). $$
For each $\varepsilon>0$, we define the perturbation function $\varphi(\tau) = (1-\varepsilon)\delta_{\tau} + \varepsilon\eta$ for each $\tau \in \Delta$. For simplicity, we write $\varphi(\tau)$ as $\ehat{\tau}$. Let $h^\varepsilon = (1-\varepsilon)h + \varepsilon\tau^*$, hence $\int_I h^\varepsilon(i) \dif\lambda(i) = \tau^*$, and
\begin{align*}
\int_\Delta u(a, \tau') \dif\ehat{\tau^*}(\tau')
    & = (1-\varepsilon)u(a, \tau^*) + \varepsilon \int_\Delta u(a, \tau')\dif\eta(\tau')    \\
    & = (1-\varepsilon)u(b, \tau^*) + \varepsilon \int_\Delta u(b, \tau')\dif\eta(\tau') = \int_\Delta u(b, \tau') \dif\ehat{\tau^*}(\tau').
\end{align*}
Then we can easily see that $h^\varepsilon$ is an $\varepsilon$-robust perfect equilibrium (associated with the perturbed player space $I$), and hence by letting $\varepsilon \to 0$, we can see that $h$ is a robust perfect equilibrium.

\paragraph{Case 2.} If $\tau_a^*$ or $\tau_b^*$ is zero, then without loss of generality we assume $\tau_b^* = 0$. Since $a$ is not weakly dominated by $b$, there exists a distribution $\tau_3$ such that $u(a, \tau_3) > u(b, \tau_3)$. Therefore, we can find a full-support distribution $\zeta$ on $\Delta$, which assigns most of the weight at $\tau_3$, such that
$$ \int_\Delta u(a, \tau') \dif\zeta(\tau')  >  \int_\Delta u(b,\tau') \dif\zeta(\tau'). $$
Let $\rho = \int_\Delta u(a, \tau') \dif\zeta(\tau') - \int_\Delta u(b,\tau') \dif\zeta(\tau') > 0$. For each $\varepsilon>0$, since $u(a, \delta_a) - u(b, \delta_a) \ge 0$ and $u$ is continuous on $\Delta$, there exists $\varepsilon' \in (0, \varepsilon)$ such that $u \bigl( a, (1-\varepsilon')\delta_a + \varepsilon'\delta_b \bigr) - u \bigl( b, (1-\varepsilon')\delta_a + \varepsilon'\delta_b \bigr) > -\frac{\varepsilon}{1-\varepsilon} \rho$.

Let $h^\varepsilon(i) \equiv (1-\varepsilon') \delta_a + \varepsilon' \delta_b$. Then induced distribution $\tau^\varepsilon = (1-\varepsilon') \delta_a + \varepsilon' \delta_b$. Let the perturbation function $\varphi(\tau) = (1-\varepsilon)\delta_{\tau} + \varepsilon\eta$ for each $\tau \in \Delta$, and we write $\varphi(\tau)$ as $\ehat{\tau}$. It is clear that
\begin{align*}
\int_\Delta u(a, \tau') \dif\ehat{\tau^\varepsilon}(\tau')
    & = (1-\varepsilon) u\bigl(a, (1-\varepsilon')\delta_a + \varepsilon'\delta_b \bigr) + \varepsilon \int_\Delta u(a, \tau') \dif\zeta(\tau')    \\
    & > (1-\varepsilon) u\bigl(b, (1-\varepsilon')\delta_a + \varepsilon'\delta_b \bigr) - \varepsilon\rho + \varepsilon \int_\Delta u(b, \tau') \dif\zeta(\tau') + \varepsilon\rho  \\
    & =  \int_\Delta u(b, \tau') \dif\ehat{\tau^\varepsilon}(\tau').
\end{align*}
While $h^\varepsilon(i)$ assigns the probability $(1-\varepsilon')$ on $a$ and $1-\varepsilon' > 1-\varepsilon$. Then $h^\varepsilon$ is an $\varepsilon$-robust perfect equilibrium (associated with the perturbed player space $I$), and hence by letting $\varepsilon \to 0$, we can see that $h$ is a robust perfect equilibrium.
\end{proof}

\small
\singlespacing

\end{document}